\documentclass[10pt]{article}

\usepackage[
  margin=1in,
  top=0.7in,
  bottom=0.8in
]{geometry}

\usepackage{iftex}
\ifPDFTeX
  \usepackage[T1]{fontenc}
  \usepackage[utf8]{inputenc}
  \usepackage{textcomp}
  \usepackage{lmodern}
\else
  \usepackage{unicode-math}
  \defaultfontfeatures{Scale=MatchLowercase}
  \defaultfontfeatures[\rmfamily]{Ligatures=TeX,Scale=1}
\fi

\usepackage{amsmath,amssymb,amsthm}

\usepackage{enumitem}

\usepackage[dvipsnames]{xcolor}

\definecolor{estfill}{HTML}{EEF3F8}
\definecolor{uqfill}{HTML}{E6F4F1}
\definecolor{uqdark}{HTML}{1F6F68}

\usepackage{
  longtable,
  booktabs,
  array,
  tabularx,
  multirow,
  makecell,
  pifont
}
\usepackage{calc}
\usepackage{etoolbox}

\newcolumntype{P}[1]{>{\raggedright\arraybackslash}p{#1}}

\newcommand{\cmark}{\ding{51}}
\newcommand{\xmark}{\ding{55}}
\newcommand{\pmark}{\raisebox{0.2ex}{\small$\triangle$}}

\makeatletter
\patchcmd{\longtable}
  {\par}
  {\if@noskipsec\mbox{}\fi\par}
  {}{}
\makeatother

\IfFileExists{footnotehyper.sty}
  {\usepackage{footnotehyper}}
  {\usepackage{footnote}}
\makesavenoteenv{longtable}

\usepackage{graphicx}

\makeatletter
\def\maxwidth{%
  \ifdim\Gin@nat@width>\linewidth
    \linewidth
  \else
    \Gin@nat@width
  \fi
}
\def\maxheight{%
  \ifdim\Gin@nat@height>\textheight
    \textheight
  \else
    \Gin@nat@height
  \fi
}
\makeatother

\setkeys{Gin}{
  width=\maxwidth,
  height=\maxheight,
  keepaspectratio
}

\usepackage{caption}
\usepackage{subcaption}
\usepackage{float}

\makeatletter
\def\fps@figure{htbp}

\floatstyle{ruled}
\newfloat{codelisting}{htbp}{lop}
\floatname{codelisting}{Listing}

\makeatother

\newtheorem{theorem}{Theorem}
\newtheorem{assumptions}{Assumptions}

\newenvironment{identifiabilityassumptions}
  {%
    \begin{assumptions}[Assumptions for Identifiability]%
  }
  {%
    \end{assumptions}%
  }

\makeatletter
\let\oldparagraph\paragraph
\renewcommand{\paragraph}{%
  \@ifstar{\articleParagraphStar}{\articleParagraphNoStar}%
}
\newcommand{\articleParagraphStar}[1]{%
  \oldparagraph*{#1}\mbox{}%
}
\newcommand{\articleParagraphNoStar}[1]{%
  \oldparagraph{#1}\mbox{}%
}

\let\oldsubparagraph\subparagraph
\renewcommand{\subparagraph}{%
  \@ifstar{\articleSubparagraphStar}{\articleSubparagraphNoStar}%
}
\newcommand{\articleSubparagraphStar}[1]{%
  \oldsubparagraph*{#1}\mbox{}%
}
\newcommand{\articleSubparagraphNoStar}[1]{%
  \oldsubparagraph{#1}\mbox{}%
}
\makeatother

\usepackage{setspace}
\usepackage{natbib}

\IfFileExists{xurl.sty}{\usepackage{xurl}}{}
\usepackage{hyperref}
\usepackage{bookmark}

\hypersetup{
  pdftitle={Hypothesis Testing for Causal Discovery},
  pdfauthor={Shreya Prakash; Fan Xia; Elena A. Erosheva},
  pdfkeywords={
    statistical guarantees,
    statistical inference,
    uncertainty quantification,
    model diagnostics,
    assumption violations
  },
  colorlinks=true,
  linkcolor=blue,
  filecolor=Maroon,
  citecolor=Blue,
  urlcolor=Blue,
  pdfcreator={LaTeX via pandoc}
}

\ifLuaTeX
  \usepackage{selnolig}
\fi

\begin{document}

\title{Statistical Inference for Bivariate Functional Causal Discovery}

\author{
Shreya Prakash$^{1}$ \and
Fan Xia$^{4}$ \and
Elena A. Erosheva$^{1,2,3}$
}
\date{%
    $^1$Department of Statistics, University of Washington, Washington, USA\\%
    $^2$School of Social Work, University of Washington, Seattle, WA, USA\\
    $^3$Center for Statistics and the Social Sciences, University of Washington, Seattle, WA, USA\\
    $^4$Department of Epidemiology and Biostatistics, University of California San Francisco, California, USA\\
    [2ex]%
}
\maketitle

\abstract{Causal discovery methods aim to determine the causal direction between variables using observational data. Functional causal discovery methods rely on structural and distributional assumptions to determine directionality but typically lack statistical inference. This paper reviews the statistical guarantees of existing functional causal discovery methods in theory, software, and applied use, highlighting a key gap: the absence of a unified inferential framework that applies broadly across model classes. As a first step toward addressing this gap, we formalize a test-based approach for bivariate causal discovery by repurposing goodness-of-fit and independence tests within a hypothesis-testing framework. Because directionality is determined through two disjoint hypothesis tests, corresponding to four causal discovery outcomes, the approach provides explicit uncertainty quantification and diagnostic insight into assumption violations. Resampling is further used to estimate the rates of causal discovery outcomes, offering an additional layer of inference. We demonstrate the use and behavior of our inferential framework through simulations that vary the degree of assumption violation, as well as through real-data applications. We conclude with practical lessons and recommendations for advancing statistical guarantees in functional causal discovery.
}

% \jnlcitation{\cname{%
% \author{Prakash M.},
% \author{Xia F}, and
% \author{Erosheva E}}.
% \ctitle{A Diagnostic Tool for Functional Causal Discovery.} \cjournal{\it Statistics in Medicine} \cvol{2024;00(00):1--18}.}

\renewcommand\thefootnote{}
\footnotetext{Preprint.}

\renewcommand\thefootnote{\fnsymbol{footnote}}
\setcounter{footnote}{1}

\section{Introduction}

Identifying and understanding causal relations is fundamental to many scientific disciplines. For example, in biology, uncovering causal links in gene regulatory networks is crucial for gaining insight into biological processes and disease mechanisms \citep{jiang2023signet}. While randomized experiments and carefully designed interventions can reveal causal relationships, these are often too expensive, too slow, or simply infeasible. As a result, reliable causal discovery methods that can extract evidence of causal structure from observational data are essential.

Causal discovery methods have been applied across the sciences \citep{glymour_review_2019, zhang2012inferring, saito2023causal, kotoku2020causal, rosenstrom2012pairwise, hu2018application}. These methods typically represent causal structure using a Directed Acyclic Graph (DAG), whose nodes correspond to random variables and whose edges represent direct causal dependencies. The acyclicity assumption rules out directed cycles and ensures that the graph encodes a well-defined causal ordering. 

In this paper, we begin by providing the first focused review of statistical guarantees in functional causal discovery, a family of causal discovery methods that identify causal direction by positing specific functional relationships between variables. Existing surveys of the literature describe model classes and algorithms, and a growing but still limited body of work has begun to develop inferential tools for causal discovery~\citep[e.g.,][]{komatsu2009,thamvitayakul2012,wang2025confidence}. In this paper, we systematically review whether and how commonly used functional causal discovery methods provide uncertainty quantification, error control, and tools for diagnosing assumption violations. Our review shows that these inferential components---central to classical statistical methodology---receive limited attention in existing functional causal discovery approaches. To address these gaps, we then propose a test-based framework for bivariate functional causal discovery.

Before turning to our focused review of statistical guarantees, we briefly review existing causal discovery methods. Two broad classes of causal discovery are constraint-based and function-based methods~\citep{glymour_review_2019}. Constraint-based approaches---such as the Fast Causal Inference (FCI) algorithm \citep{spirtes2000constructing} and the Peter–Clark (PC) algorithm \citep{spirtes2000constructing}---search for patterns of conditional independence in the data. Their validity depends on the faithfulness assumption, which requires that all conditional independence relations present in the observed distribution are implied by the DAG. In general, constraint-based methods identify a Markov equivalence class rather than a unique DAG. Constraint-based methods work best with large sample sizes, where tests of conditional independence have sufficient power \citep{glymour_review_2019}, and they provide no information in purely bivariate settings where no conditional independences exist.

Functional causal discovery methods take a different approach: they use structural assumptions on the functional form of the relationship between variables to identify the causal direction. In the bivariate case, for two random variables $X$ and $Y$, suppose the true causal direction is $X$ causes $Y$. Then, functional causal models are of the form:
$$Y = f(X,\eta, \mathbf{\phi}),$$ 
where $\eta$ is noise that is independent of $X$, $f \in \mathcal{F}$ is the functional form, and $\phi$ is the parameter set. Given $\boldsymbol{\phi}$, some models additionally assume that $f$ is invertible so that $\eta$ can be recovered uniquely from $(X, Y)$. These methods identify the causal direction by exploiting asymmetries: when the model assumptions hold, the independence between the cause $X$ and the noise $\eta$ appears only in the true direction, not the reverse. 

Prominent examples include the Linear Non-Gaussian Acyclic Model (LiNGAM) \citep{shimizu_use_2008}, where $Y = X + \eta$ and at most one of $X$ or $\eta$ is Gaussian; the Additive Noise Model (ANM) \citep{hoyer2009nonlinear}, where $Y = f_{\mathrm{AN}}(X) + \eta$; and the Post-Nonlinear Model (PNL) \citep{zhang2006extensions, zhang2012identifiability}, where $Y = f_2(f_1(X) + \eta)$ with $f_2$ invertible. \citet{shimizu2006linear, shimizu2011directlingam}, \citet{hoyer2009nonlinear}, and \citet{zhang2012identifiability} formalize identifiability results for LiNGAM, ANM, and PNL, respectively, and provide algorithms for estimating the causal direction.

Functional causal discovery methods have been used widely in neuroscience, epidemiology, psychology, genetics, and other domains in the medical and social sciences. For example,~\citet{saito2023causal} used LiNGAM to study factors contributing to nitrogen oxide generation in a coal-fired power plant.~\citet{rosenstrom2012pairwise} used LiNGAM to examine causal relationships between sleep and depression, while~\citet{hu2018application} applied ANMs for causal discovery in genetic studies.~\citet{song2017tell} applied ANM and PNL models to diverse bivariate datasets across geography, biology, physics, and economics.

Despite their widespread use, functional causal discovery methods remain limited by the absence of statistical guarantees---uncertainty quantification, error control, and consistency properties that form the basis of statistical inference \citep[e.g.,][]{wasserman2010statistics}. Our focus is on frequentist uncertainty quantification and inferential guarantees for functional causal discovery. Bayesian model-selection approaches provide posterior support for causal structures or directions under specified probabilistic models and priors \citep{hoyer2009bayesian, stegle2010probabilistic, dhir2024bivariate}; these posterior quantities represent a distinct inferential target from the frequentist guarantees and diagnostic assessment of assumption violations pursued here. In particular, many causal discovery methods lack accompanying frequentist inferential frameworks: they rely on deterministic outcomes by returning a single, potentially misleading point estimate of the causal direction. In addition, functional causal discovery depends on strong assumptions---such as specific functional forms---to identify the causal model from population information \citep[e.g.,][]{shimizu2006linear, hoyer2009nonlinear}. As these assumptions are likely to be violated in practice to various degrees, there is a clear need for inferential tools that quantify uncertainty and provide information about the extent of possible assumption violations.

In Section 2, we review the state of statistical guarantees provided by existing functional causal discovery methods. In Section 3, we introduce a principled way to assess causal discovery evidence and to recognize when the data do not support a definitive causal direction. We do so by repurposing goodness-of-fit and independence tests into a formal hypothesis-testing procedure for causal direction, which we term the \textit{test-based} approach for bivariate functional causal discovery. In Section 4, we demonstrate the use and behavior of the test-based approach through simulations that vary the degree of assumption violation, as well as through real-data applications. Finally, in Section 5, we situate the test-based approach within the broader causal discovery literature, offer practical guidance for applied researchers, and discuss directions for future work on strengthening statistical guarantees in functional causal discovery.

\section{Current State of Statistical Guarantees in Functional Causal Discovery}\label{sec-stat-guarantees}

\begin{table*}[htb]
\tiny
\setlength\tabcolsep{4pt}
\caption{Current state of statistical guarantees (column blocks) in functional causal discovery by literature component (row blocks). Our proposed bivariate test-based approach is boxed for emphasis. Model class coverage for each paper is given in column Model. A check (\cmark) indicates the criterion is explicitly supported; a cross (\xmark) indicates it is not; a triangle (\pmark) indicates partial or limited support. NA indicates that a criterion is not applicable.}
\label{tab:fcd-inference-summary}
\begin{tabularx}{\textwidth}{l *{10}{c}}
\toprule
& \multicolumn{4}{c}{\textbf{Inference}} & \multicolumn{2}{c}{\textbf{Decision}} & \multicolumn{1}{c}{\textbf{Model}} \\
\cmidrule(lr){2-5}\cmidrule(lr){6-7}\cmidrule(lr){8-8}
\makecell[l]{\textbf{Method / Paper}} &
\makecell[c]{Any\\inferential\\quantity} &
\makecell[c]{Formal\\inferential\\framework} &
\makecell[c]{Multiple\\forms\\of inference} &
\makecell[c]{Impact of\\assumption\\violations\\understood} &
\makecell[c]{Deterministic\\decision\\not enforced} &
\makecell[c]{Allows\\inconclusive\\outcome} &
\makecell[c]{Model\\class\\coverage} \\
\midrule
\textbf{Core Model Families} \\
\quad ~\cite{shimizu2011directlingam} &
\xmark & \xmark & \xmark & \xmark & \xmark & \xmark &  LiNGAM \\
\quad ~\cite{hoyer2009nonlinear} &
\cmark & \xmark & \xmark & \pmark & \cmark & \cmark & ANM \\
\quad ~\cite{peters2014causal} &
\cmark & \xmark & \xmark & \pmark & \pmark & \pmark & ANM \\
\quad ~\cite{zhang2012identifiability} &
\cmark & \xmark & \xmark & \pmark & \cmark & \cmark & PNL \\
\addlinespace[2pt]
\textbf{Software} \\
\quad ~\cite{Ikeuchi2023-lingam} &
\cmark & \xmark & \xmark & \xmark & \pmark & \pmark & LiNGAM\\
\quad ~\cite{Kalisch2012-pcalg} &
\xmark & \xmark & \xmark & \xmark & \xmark & \xmark &  LiNGAM \\
\quad ~\cite{Kalainathan2019-cdt} &
\xmark & \xmark & \xmark & \xmark & \xmark & \xmark & ANM \\
\quad ~\cite{Zheng2024-causallearn} &
\cmark & \xmark & \xmark & \xmark & \pmark & \pmark & LiNGAM/ANM/PNL\\
\addlinespace[2pt]
\textbf{Benchmarks/Applications} \\
\quad ~\cite{mooij_distinguishing_2016} &
\cmark & \xmark & \xmark & \pmark & \pmark & \pmark & ANM \\
\quad ~\cite{rosenstrom2012pairwise} &
\cmark & \cmark & \xmark & \xmark & \xmark & \xmark & LiNGAM \\
\quad ~\cite{motokawa2020causal} &
\cmark & \cmark & \xmark & \xmark & \xmark & \xmark & LiNGAM \\
\quad ~\cite{jiao2018bivariate} &
\cmark & \xmark & \xmark & \pmark & \pmark & \pmark & ANM \\
\quad ~\cite{hu2018application} &
\cmark & \xmark & \xmark & \pmark & \pmark & \pmark & ANM \\
\quad ~\cite{song2017tell} &
\cmark & \xmark & \xmark & \xmark & \xmark & \xmark & ANM/PNL \\
\addlinespace[2pt]
\textbf{Inferential Methods} \\
\quad ~\cite{komatsu2009} &
\cmark & \cmark & \xmark & \xmark & \xmark & \xmark & LiNGAM \\
\quad ~\cite{thamvitayakul2012} &
\cmark & \cmark & \xmark & \xmark & \xmark & \xmark & LiNGAM \\
\quad ~\cite{schkoda2025goodness} &
\cmark & \cmark & \xmark & \cmark & NA & NA & LiNGAM \\
\quad ~\cite{wang2025confidence} &
\cmark & \cmark & \xmark & \cmark & \cmark & \cmark & LiNGAM/ANM \\
\addlinespace[2pt]
\quad \fbox{\textbf{Test-based Approach (proposed)}} &
\cmark & \cmark & \cmark & \cmark & \cmark & \cmark & LiNGAM/ANM/PNL \\
\bottomrule
\end{tabularx}

\vspace{4pt}
\footnotesize
\textit{Notes.} ``Any inferential quantity'' includes $p$-values, bootstrap frequencies, or confidence sets. ``Formal inferential framework'' means that these quantities are developed and interpreted within a formalized inferential framework, such as a hypothesis-testing or confidence-set procedure. ``Multiple forms of inference'' indicates that multiple complementary inferential tools are integrated within a coherent framework. % as recommended for reliable inference by \citet{wasserstein2019beyond}.
All other columns are self-explanatory.
\end{table*}

To simplify the exposition, we view current functional causal discovery literature as comprising several distinct components. The first group are papers presenting core methodological frameworks such as the Linear Non-Gaussian Acyclic Model (LiNGAM), Additive Noise Model (ANM), and Post-Nonlinear (PNL) model~\citep{shimizu2006linear,peters2014causal,zhang2012identifiability}. The second group is comprised of papers that focus on software implementations of these methods, providing practical tools for applied researchers~\citep{Ikeuchi2023-lingam,Kalisch2012-pcalg,Kalainathan2019-cdt,Zheng2024-causallearn}. The third group is a set of benchmarking and applied studies that use these methods on real datasets to assess their empirical performance~\citep{mooij_distinguishing_2016,rosenstrom2012pairwise,motokawa2020causal,jiao2018bivariate,hu2018application,song2017tell}. The forth group are papers that aim to establish formal statistical guarantees for causal discovery~\citep{komatsu2009,thamvitayakul2012,schkoda2025goodness,wang2025confidence}. We consider the model class each paper addresses to understand the extent to which these statistical guarantees apply within the broader domain of functional causal discovery.

The papers in Table~\ref{tab:fcd-inference-summary} are intended to be representative rather than exhaustive lists, with the level of coverage varying by category: the methodological papers highlight the core model families, the software papers reflect widely used implementations, the applied studies illustrate case studies in the literature, and the inferential papers largely capture the existing formal work. The papers also vary in dimensional scope. Several focus on bivariate causal discovery, including \citet{mooij_distinguishing_2016}, \citet{rosenstrom2012pairwise}, \citet{jiao2018bivariate}, \citet{hu2018application}, and \citet{song2017tell}, as does our proposed approach. \citet{hoyer2009nonlinear} and \citet{zhang2012identifiability} develop their main identifiability results in the bivariate setting but also consider multivariate extensions, while the remaining papers address multivariate settings directly. Notably, although the framework of \citet{schkoda2025goodness} is multivariate, they illustrate their goodness-of-fit tests using bivariate cause–effect pairs.

We discuss statistical guarantees by distinguishing whether the methodology relies on a clearly defined inferential quantity---such as a $p$-value, confidence interval, or bootstrap frequency---and has a well-specified inferential framework---such as a formal hypothesis testing or a bootstrap framework---that situates the inferential quantity within a broader statistical inference structure.

The importance of clearly defined inferential quantities, coherent inferential frameworks and multiple forms of inference echoes broader guidance for statistical practice~\citep{wasserstein2019beyond, benjamini2021asa}. 
%The American Statistician’s statements on $p$-values~\citep{wasserstein2016asa, wasserstein2019beyond} emphasize that reliable inference is achieved by integration of multiple forms of inference within a framework that makes assumptions and uncertainty quantification explicit, rather than by reliance on ``bright-line rules for justifying scientific claims or conclusions'' which ``can lead to erroneous beliefs and poor decision making''~\citep[p.~2]{wasserstein2019beyond}. 
In the context of causal discovery, reliance on a single inferential device---such as a test statistic or a confidence set---may therefore limit interpretability and robustness. Accordingly, we assess whether each method incorporates multiple forms of inference as recommended for best practices.

%We further distinguish between methods that always enforce a deterministic conclusion---typically by comparing a single numerical summary such as a test statistic---and methods that allow for multiple forms of inference. 
%The ASA’s Statement on $p$-values~\citep{wasserstein2016asa} emphasizes that ``[b]y itself, a $p$-value does not provide a good measure of evidence regarding a model or hypothesis,'' underscoring the broader principle that inferential summaries should not be used as automatic decision rules in isolation. 
Allowing for inconclusive outcomes may particularly be relevant to causal discovery.~\citet{genin2020statistical,genin2024success} show that, although uniformly consistent procedures for orienting causal relationships may not exist under common modeling assumptions, weaker and practically meaningful guarantees remain attainable: causal directions may be statistically decidable or progressively solvable, with procedures permitted to remain inconclusive or to retract earlier conclusions as evidence accumulates. From this perspective, methods that accommodate inconclusive results more faithfully reflect the uncertainty inherent in statistical inference, whereas deterministic rules risk overstating the evidence when assumptions are violated and/or when the signal is weak.

The core methodological frameworks (LiNGAM/ANM/PNL) provide population-level identifiability and, for some procedures, asymptotic consistency guarantees~\citep{shimizu2006linear,peters2014causal,zhang2012identifiability,mooij_distinguishing_2016}. However, these results rely on assumptions that may often be violated in practice. Beyond asymptotic consistency, \citet{wang2020high} provide finite-sample guarantees for causal discovery in high-dimensional LiNGAMs under sparsity and distributional regularity conditions, allowing the number of variables to grow faster than the sample size. These guarantees nevertheless remain conditional on the assumed model and do not protect against violations of its underlying assumptions. The core LiNGAM papers~\citep{shimizu2006linear,shimizu2011directlingam} do not define a formal inferential quantity or framework for quantifying uncertainty in the selected causal direction; instead, they enforce a directional decision by comparing dependence or independence measures---such as mutual information---used as comparative ``scores.''

The use of statistical tests to assess residual–predictor independence, and to permit inconclusive causal discovery outcomes when both directions are rejected or fail to reject, has been discussed in the ANM and PNL contexts~\citep{hoyer2009nonlinear,zhang2012identifiability}. However, these discussions remain largely conceptual: the testing framework itself is not formally developed, and the meaning of each test outcome---particularly how rejection or non-rejection should be interpreted for causal direction identification---is not clearly specified. This lack of a well-defined inferential structure has led to heuristic uses of $p$-values in several studies and software implementations, such as selecting the direction with the larger $p$-value~\citep{peters2014causal}, rather than interpreting the $p$-values within a principled statistical testing framework.

Overall, none of the core model family (LiNGAM/ANM/PNL) papers we reviewed formalize an inferential framework or incorporate multiple complementary inferential tools. They tend to either enforce a deterministic directional decision~\citep{shimizu2006linear,shimizu2011directlingam,peters2014causal} or, at most, contain limited discussion of assumption violations and potential inconclusive outcomes (ANM/PNL)~\citep{hoyer2009nonlinear,zhang2012identifiability,peters2014causal}.

Existing software implementations~\citep[e.g.,][]{Kalisch2012-pcalg,Kalainathan2019-cdt} often follow a similar sequence of steps. That is, for a given a bivariate dataset of size $N$ on a pair of variables $X$ and $Y$, models are first fitted separately in each direction ($X \to Y$ and $Y \to X$); second, a dependence measure between the predictor and residuals (e.g., mutual information or Hilbert Schmidt Independence Criterion (HSIC)~\citep{gretton_kernel_2008}) is computed; and finally, the direction with the smaller dependence is taken as ``causal.'' These procedures provide no quantification of uncertainty. While some implementations report per-direction $p$-values~\citep{Zheng2024-causallearn,Ikeuchi2023-lingam}, it is left to the practitioner whether to interpret them inferentially or simply use them as comparative scores, as recommended by~\citet{peters2014causal}.

% Turning to the final group of 

Turning to benchmarking and applied studies, we examine how these works handle inference and decision-making in practice. Because the core model families and their software implementations provide limited support for formal statistical inference, similar limitations generally appear in their applications. For example, \citet{mooij_distinguishing_2016} present the 103 cause–effect benchmark pairs for causal discovery and apply ANM-based algorithms to these pairs. They provide some discussion of $p$-values and inconclusive outcomes but do not describe any formal inferential framework. In particular, they discuss using the HSIC independence measure as a heuristic score, selecting the model with the smaller HSIC estimate, rather than interpreting the corresponding $p$-values within a principled inferential structure \citep{mooij_distinguishing_2016}. 

Some applications use inferential tools to quantify uncertainty; however, these papers inherit the same limitations as the underlying inferential procedures. For example, several applications of LiNGAM-based methods exist in practice~\citep[e.g.,][]{saito2023causal,kotoku2020causal,rosenstrom2012pairwise,motokawa2020causal}, and some of these employ bootstrap methods to quantify uncertainty in the estimated causal direction~\citep[e.g.,][]{motokawa2020causal,rosenstrom2012pairwise}. As we discuss in the following section, these bootstrap results can be misleading when the assumptions of the underlying LiNGAM model are violated. Additionally, some applications of ANM-based methods use permutation tests to obtain $p$-values~\citep[e.g.,][]{jiao2018bivariate,hu2018application}. In these approaches, permutation is used to approximate a null distribution under which the variables are independent or no causal direction is present. The observed relationship between the variables is broken while their marginal distributions are preserved. Specifically, \citet{jiao2018bivariate} hold one variable fixed and permute the other, whereas \citet{hu2018application} permute $X$ and $Y$. For each permuted dataset, \citet{jiao2018bivariate} refit the directional models and recompute the test statistic, while \citet{hu2018application} recompute the distance-correlation statistic. These statistics form empirical null distributions against which the observed statistics are compared to obtain $p$-values. However, these works do not explicitly discuss how each decision outcome should be interpreted within a formal hypothesis-testing framework, and hypothesis testing for causal discovery remains largely informal and often only implicitly addressed~\citep{hoyer2009nonlinear}. Other applications~\citep[e.g.,][]{song2017tell} report $p$-values but rely on direct $p$-value comparisons to determine the favored direction.

The final group of papers comprises methods developed specifically to provide statistical inference for functional causal discovery (e.g.,~\citealp{komatsu2009,thamvitayakul2012,schkoda2025goodness,wang2025confidence}). Work in this area remains limited and
%There is a paucity of work aimed at building inferential tools for causal discovery (e.g.,~\citealp{komatsu2009,thamvitayakul2012,wang2025confidence,schkoda2025goodness}). The existing work in this area 
can largely be grouped into three categories: methods that estimate bootstrap-type reliability, those that construct confidence sets for causal structures, and goodness-of-fit tests for particular causal model classes.

\citet{komatsu2009} extends the LiNGAM framework with a multiscale bootstrap procedure to assess the stability of estimated causal orderings, which is implemented in the \texttt{lingam} Python package~\citep{Ikeuchi2023-lingam}. Their algorithm repeatedly refits LiNGAM on bootstrap samples, computes the mutual information between residuals as a dependence measure, and estimates the bias-corrected selection probability with which the same ordering is recovered. The resulting values represent approximately unbiased bootstrap probabilities that indicate how consistently LiNGAM selects a given causal ordering across resampled datasets. Similarly, \citet{thamvitayakul2012} adopt a resampling-based approach within the Direct-LiNGAM algorithm to estimate the uncertainty of directional signals, quantifying how frequently a particular edge orientation appears across bootstrap samples. While these methods represent important early steps toward incorporating uncertainty into causal discovery, they remain limited: the resulting probabilities depend on the assumed LiNGAM model class and can be misleading when the underlying assumptions are violated.

%Another recent contribution to statistical inference in causal discovery is that of 
\citet{schkoda2025goodness} develops formal goodness-of-fit tests for assessing whether an observed distribution is compatible with a linear non-Gaussian structural equation model, with or without latent confounding. Rather than determining causal direction, their approach provides diagnostic information about the adequacy of the LiNGAM model class and can therefore complement methods for causal direction estimation. \citet{schkoda2025goodness} apply their method to the Tübingen cause--effect pairs~\citep{mooij_distinguishing_2016} to ``offer a classification of the cause-effect pairs into a group for which a linear model without confounding is tenable, a group for which a linear model is tenable after inclusion of a single confounder and a group that may be best analyzed using nonlinear methods.''

\citet{wang2025confidence} provides the most general inferential frameworks currently available for uncertainty quantification in causal discovery. Their method applies to identifiable structural equation models with additive errors, encompassing both linear and nonlinear additive-noise models. \citet{wang2025confidence} construct a confidence set of causal orderings by inverting a goodness-of-fit test. Thus, for a bivariate case, the set may contain one direction, neither, or both, reflecting whether the data provide sufficient evidence to distinguish the models. When the set contains neither direction, the procedure offers some insight into model-class misspecification, as this indicates that the ``model class does not capture the data-generating process''~\citep[p.~2]{wang2025confidence}. Although confidence sets may carry diagnostic information---for example, the set that contains both directions indicates that the data do not distinguish between them---the paper does not discuss or formally develop a diagnostic framework. Conceptually, this approach parallels a hypothesis-testing framework, since confidence sets and tests are dual procedures, and does not incorporate another complementary form of inference. In addition, it does not currently extend to the PNL model class where noise enters non-additively.
% However, this discussion is limited to that specific case, and further examination of other types of assumption violations could broaden the framework’s interpretability. Conceptually, this approach parallels a hypothesis-testing framework, since confidence sets and tests are dual procedures. However, it does not currently extend to the PNL model class, where noise enters non-additively. 

Based on our review, we identify a clear need for a formal inferential framework that can be paired with any functional causal discovery model class (LiNGAM, ANM, or PNL). Such a framework should integrate multiple complementary inferential tools within a coherent structure that makes both assumptions and uncertainty explicit. In this paper, we aim to address this gap through our proposed \emph{test-based approach}.

\section{A Hypothesis-Testing Framework for Bivariate Functional Causal Discovery}\label{sec-meth}

In this section, we develop a framework for functional causal discovery that is grounded in classical hypothesis testing. This framework, which we term the \textit{test-based approach}, repurposes standard goodness-of-fit and independence tests into a formal procedure for causal direction detection and inference. As summarized in Table~\ref{tab:fcd-inference-summary}, the test-based approach provides a formal inferential framework that integrates multiple complementary forms of inference and includes an explicit ``inconclusive'' outcome whose behavior under assumption violations is well understood. It can be paired with any major functional causal discovery model class, including PNL and its special cases, LiNGAM and ANM. These aspects place it in a favorable position within the broader literature in terms of its statistical guarantees.

Given a bivariate dataset of size $N$ on a pair of variables $X$ and $Y$, the test-based framework infers directionality through setting up two hypothesis tests, one for each direction, each conducted with controlled Type~I error. In addition to the Type~I error control provided by the underlying hypothesis tests, we introduce a resampling-based procedure that estimates and constructs confidence intervals for the frequencies with which the test-based approach yields each causal discovery outcome---(a) rejecting the tests for both directions $X \to Y$ and $Y \to X$, (b) failing to reject both tests, (c) rejecting only the test for $X \to Y$, or (d) rejecting only the test for $Y \to X$. This procedure adds an additional layer of uncertainty quantification beyond the formal guarantees of the tests themselves. Consequently, the test-based approach provides formal statistical guarantees for functional causal discovery, whereas most existing methods yield only deterministic decisions without any accompanying assessment of uncertainty (e.g., \cite{peters2014causal, shimizu2011directlingam}). We illustrate the methodology when it is paired with the bivariate Post-Nonlinear (PNL) model class, which encompasses the Additive Noise Model (ANM) and LiNGAM as special cases.

\subsection{Setup}

Suppose we have a dataset of size $N$ on a pair of variables $X$ and $Y$, where a causal direction exists. In bivariate causal discovery, the task is to distinguish between $X \rightarrow Y$ and $Y \rightarrow X$. Under Assumptions~\ref{assump:model}, the causal direction is identified at the population level under the PNL model class~\cite{zhang2012identifiability}.

\begin{identifiabilityassumptions}\label{assump:model} \:\
\begin{enumerate}[label=(A\arabic*)]
    \item The data-generating mechanism is correctly specified within the assumed functional model class.
    \item $X$ and $Y$ have an acyclic relationship.
    \item $X$ and $Y$ are unconfounded, meaning they share no common causes.
    \item The data are independent and identically distributed (i.i.d.).
    \item None of the five unidentifiable cases characterized by \citet{zhang2012identifiability}
    (summarized in Appendix) occur; these include,
    for example, situations where the relationship between $X$ and $Y$ is linear and all random
    variables are Gaussian.
\end{enumerate}
\end{identifiabilityassumptions}

\begin{figure}[htb]
    \centering
    \includegraphics[width=\linewidth]{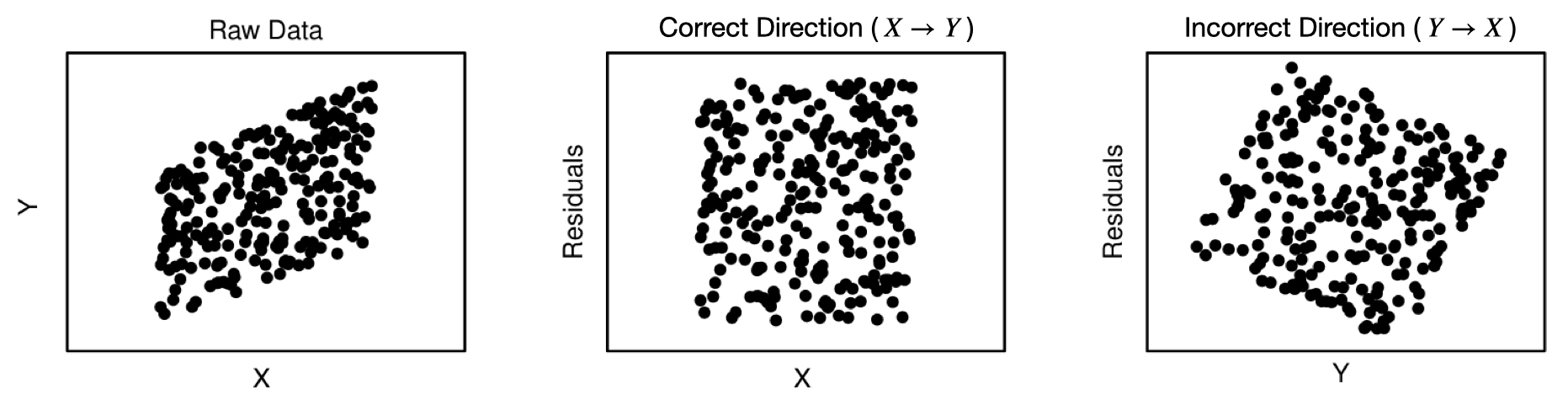}
    \caption{Left: Raw data. Middle: Residuals from regressing the outcome on the predictor in the correct causal direction. Right: Residuals from regressing the predictor on the outcome in the incorrect direction. Adapted from \citet{wang2025confidence}.}
    \label{fig:asymmetry_example}
\end{figure}

Assumption A1 makes explicit a condition that is implicitly required when applying population-level identifiability results in practice. Identifiability theory assumes knowledge of the true functional model and asks whether the resulting joint distribution admits that model representation in only one causal direction. In finite-sample applications, however, the fitted model class must adequately represent the data-generating mechanism for the identifiability result to inform estimation of the causal direction. We therefore state A1 explicitly because model misspecification can affect the inferred direction.

When Assumptions \ref{assump:model} are satisfied, the noise term in the true causal direction is independent of the predictor but not the other way around. This property can be exploited to infer the causal direction. To make this dependence structure explicit, assume the true data-generating mechanism is given by:

\begin{equation}
    Y = f_2(f_1(X) + \eta),
    \label{eq:gen_y}
\end{equation}

where $f_1$ and $f_2$ are the generating functions, and $f_2$ is invertible. The noise term $\eta$ is, by construction, independent of the predictor $X$.  

Now, suppose $f_1$ is also invertible. Then we can write:
\begin{equation}
    X = f_1^{-1}(f_2^{-1}(Y)-\eta).
    \label{eq:X_def}
\end{equation}

Next, consider attempting to represent the data-generating mechanism in the reverse direction, $Y \to X$. Specifically, we consider representations of the form
\begin{equation}
X = g_2\!\big(g_1(Y) + \xi\big),
\label{eq:reverse_model}
\end{equation}
where $g_1$ and $g_2$ are unknown functions and $g_2$ is assumed to be invertible. Let $\mathcal{G}_1$ and $\mathcal{G}_2$ denote the function classes from which $g_1$ and $g_2$ are chosen (e.g., linear functions under LiNGAM).

To select a particular representation within this class, we define $(g_1^\star, g_2^\star)$ as population risk minimizers under a loss function $\ell$:
\begin{equation}
(g_1^\star,g_2^\star) \in 
\arg\min_{g_1\in\mathcal{G}_1,\; g_2\in\mathcal{G}_2}
\mathbb{E}\!\left[\ell\!\left(g_2^{-1}(X)-g_1(Y)\right)\right].
\label{eq:reverse_pop_risk}
\end{equation}

Given $(g_1^\star,g_2^\star)$, we define the reverse-direction residual as
\begin{equation}
\xi := g_2^{\star -1}(X)-g_1^\star(Y),
\label{eq:xi_def}
\end{equation}
so that \eqref{eq:reverse_model} holds by construction.

The functions $(g_1^\star,g_2^\star)$ represent the best-fitting functions within the chosen model classes when attempting to express the data-generating process in the reverse direction, and depend on both the model class and the loss function $\ell$.
They are introduced to illustrate how dependence arises when fitting the incorrect causal direction.
According to \eqref{eq:X_def}, the residual $\xi$ can be written as a deterministic function of $(Y,\eta)$. Since $Y$ itself is generated from $\eta$ via \eqref{eq:gen_y}, we generally have $Y \not\perp \xi$.
Thus, if $X \to Y$ is the true data-generating mechanism, then $X \perp \eta$ but typically $Y \not\perp \xi$, except in the well-known unidentifiable cases under Assumption~A5 (e.g., the linear-Gaussian setting).

Figure~\ref{fig:asymmetry_example} illustrates this asymmetry property between the correct and incorrect causal directions. Additional intuition for how $\xi$ behaves under different properties (e.g., additive $f_1$) is provided in the Appendix.

As shown above, the PNL model relies on independence between residuals and predictors to infer the causal direction---a property that holds only when the fitted model is correctly specified. Thus, assessing the causal direction naturally involves evaluating both model adequacy (goodness-of-fit) and independence, recasting causal discovery as a joint problem of goodness-of-fit and independence testing. Building on this observation, the next section introduces the \textit{test-based approach}, a hypothesis-testing framework for statistical inference in functional causal discovery.

\subsection{Test-based Approach}

The test-based approach translates the causal discovery task of selecting between the directions $X \rightarrow Y$ and $Y \rightarrow X$ into the following pair of hypotheses:

\begin{equation}
\label{eq:h_arrow}
\begin{aligned}
H_{X\to Y}^0 &:~ X \rightarrow Y, \quad &H_{X\to Y}^1 &:~ \text{otherwise},\\[4pt]
H_{Y\to X}^0 &:~ Y \rightarrow X, \quad &H_{Y\to X}^1 &:~ \text{otherwise.}
\end{aligned}
\end{equation}

Under the key assumptions of PNL causal discovery \eqref{assump:model}, the noise is independent of the predictor in the true causal direction. Thus, the pair of hypotheses from~\eqref{eq:h_arrow} can be reformulated to test both independence and goodness of fit in each direction. Specifically, for prespecified function classes $\mathcal{F}_1, \mathcal{F}_2, \mathcal{G}_1, \mathcal{G}_2$, we consider:

\begingroup
\footnotesize
\begin{equation}
\label{eq:h_indep}
\begin{aligned}
H_{X\to Y}^0: \exists\, f_1\in \mathcal{F}_1,\ f_2 \in \mathcal{F}_2 \ \text{s.t.}\ Y = f_2(f_1(X) + \eta),\ f_2 \ \text{invertible},\ X \perp \eta, \:\ H_{X\to Y}^1: otherwise; \\[0.5em]
        H_{Y\to X}^0: \exists\, g_1 \in \mathcal{G}_1,\ g_2 \in \mathcal{G}_2 \ \text{s.t.}\ X = g_2(g_1(Y) + \xi),\ g_2 \ \text{invertible},\ Y \perp \xi, \:\ H_{Y\to X}^1: otherwise.
\end{aligned}
\end{equation}
\endgroup

Here, the function classes $\mathcal{F}_1, \mathcal{F}_2, \mathcal{G}_1, \mathcal{G}_2$ can encode modeling assumptions (e.g. linearity) and are typically taken to be the same across directions (i.e., $\mathcal{F}_1 = \mathcal{G}_1, \mathcal{F}_2 = \mathcal{G}_2$). One can choose among different independence tests to conduct the hypothesis testing in~\eqref{eq:h_indep}. For example, we use the test introduced by \citet{sen_testing_2014}, which simultaneously assesses (i) independence between the residuals and the predictor using the Hilbert–Schmidt independence criterion and (ii) the goodness of fit within the specified function classes. Although this test was developed for additive noise models, we use it in our simulations for PNL models. Developing general formal theoretical guarantees for this application to PNL models is beyond the scope of this paper. Nevertheless, the numerical results in Section~\ref{sec-sims} and the Appendix provide empirical evidence that, under the PNL misspecification settings considered, the test rejects both causal directions due to model misspecification, including when the inner function, the outer function, or both are misspecified.

When incorporated into a hypothesis-testing framework, the test-based approach provides a principled way to assess whether a proposed directional model is compatible with the data, in contrast to methods that rely solely on comparing dependence measures (scores). This enables detection of situations where model assumptions are not satisfied, leading to inconclusive outcomes.

% \begin{table}[htb]
%     \centering
%     \caption{Causal discovery outcomes and their interpretation when using hypothesis tests in \eqref{eq:h_indep} to test for goodness-of-fit and independence test}
%     \begin{tabular}{p{6cm}|p{9.5cm}}
%     \hline
%          \textbf{Test Outcome} & \textbf{Interpretation} \\
%          \hline
%          {Reject both $H_{X\to Y}^0$ and $H_{Y\to X}^0$.} & {\textit{Inconclusive} outcome; can arise from model misspecfication. %The null $H_0$ (\ref{eq:h0}) would not hold in both directions in this case.
%          }
%          \\
%          \hline
%          {Fail to reject both $H_{X\to Y}^0$ and $H_{Y\to X}^0$.} & {\textit{Inconclusive} outcome; can arise from identifiability issues or small sample size.
%          }\\
%          \hline
%          {Reject $H_{X\to Y}^0$ and fail to reject $H_{Y\to X}^0$.} & {Favors the direction $Y \rightarrow X$.}\\
%          \hline
%          {Fail to reject $H_{X\to Y}^0$ and reject $H_{Y\to X}^0$.} & {Favors the direction $X \rightarrow Y$.}\\
%          \hline
%     \end{tabular}
%     \label{Tab:sen_assume}
% \end{table}

\begin{table}[htb]
    \centering
    \small
    \caption{Causal discovery outcomes and their interpretation when using hypothesis tests in \eqref{eq:h_indep} to test for goodness-of-fit and independence test}
    \begin{tabular}{p{4cm}|p{5cm}|p{6cm}}
    \hline
    & \textbf{Reject $H_{Y\to X}^0$} & \textbf{Fail to reject $H_{Y\to X}$} \\
    \hline
    \textbf{Reject $H_{X\to Y}^0$} 
    & Inconclusive; potential model misspecification. 
    & Favors $Y \rightarrow X$ \\
    \hline
    \textbf{Fail to reject $H_{X\to Y}^0$} 
    & Favors $X \rightarrow Y$ 
    & Inconclusive; potential identifiability issues or small sample size. \\
    \hline
    \end{tabular}
    \label{Tab:sen_assume}
\end{table}

Table~\ref{Tab:sen_assume} summarizes how different assumption violations correspond to each possible outcome under the test-based approach. Rejecting $H_{X\to Y}^0$ while failing to reject $H_{Y\to X}^0$ favors the direction $X \to Y$, whereas rejecting $H_{Y\to X}^0$ while failing to reject $H_{X\to Y}^0$ favors $Y \to X$. When both null hypotheses are rejected, the outcome is inconclusive, which occurs whenever neither causal direction provides an adequate representation of the data under the assumed model families. A common reason for this is model misspecification (i.e., violation of Assumption~A1). While the Post-Nonlinear (PNL) model class is flexible enough to represent a wide range of data-generating mechanisms, practical implementations rely on estimated functions (e.g.,~$\hat{f}_1, \hat{f}_2$) that may not adequately capture the true relationships. As a result, the test-based procedure can still detect apparent model misspecification when sample size is limited or when the functional form lacks sufficient flexibility to approximate the true data-generating process. Moreover, for additive-noise models (ANMs) and LiNGAM, which are special cases of PNL, both null hypotheses may be rejected if the true data-generating process lies outside those restricted model classes (e.g., when the errors are not additive or the model is not linear, respectively). Conversely, when neither null hypothesis is rejected, the outcome is also inconclusive, reflecting situations in which the available data do not provide sufficient statistical power to distinguish between directions under the assumed model families. Inconclusive outcomes can arise from small sample sizes or identifiability issues. In the LiNGAM setting, for instance, this occurs under Gaussianity, where the model becomes unidentifiable \citep{shimizu_use_2008}.

In addition to inferential guarantees provided by hypothesis testing, we recommend incorporating a resampling-based procedure that estimates the frequency of each causal discovery outcome (Table~\ref{Tab:sen_assume}) and constructs associated confidence intervals for the test-based approach. For example, if one were to resample with replacement $S$ times, the rate of favoring the direction $X \to Y$ would be 

\begin{equation}
\label{eq:outcome_rate}
    \frac{1}{S}\sum_{s=1}^S \mathbf{I}(\text{Reject $H_{Y\to X}^0$ and fail to reject $H_{X\to Y}^0$ for replication $s$}).
\end{equation}

\begin{theorem}[Consistency and confidence intervals for bootstrap outcome rates]
\label{thm:bootstrap_outcome_rate}
Let \(\widehat p_{X\to Y,S}\) denote the outcome rate in Equation \ref{eq:outcome_rate}, and define the conditional bootstrap outcome probability
\[
p_{X\to Y,N}^{*}
=
\mathbb{P}^{*}\!\left(
\text{Reject }H_{Y\to X}^{0}
\text{ and fail to reject }H_{X\to Y}^{0}
\,\middle|\,D_N
\right),
\]
where \(\mathbb{P}^{*}\) denotes probability under the bootstrap distribution conditional on the observed dataset \(D_N\).

If the \(S\) bootstrap samples are drawn independently and the same test-based procedure is applied to each sample, then, conditional on \(D_N\),
\[
\widehat p_{X\to Y,S}
\xrightarrow{\mathrm{a.s.}}
p_{X\to Y,N}^{*}
\qquad\text{as }S\to\infty.
\]
Moreover, if \(0<p_{X\to Y,N}^{*}<1\), then
\[
\sqrt{S}\left(
\widehat p_{X\to Y,S}-p_{X\to Y,N}^{*}
\right)
\xrightarrow{d}
\mathcal{N}\!\left(
0,\,
p_{X\to Y,N}^{*}
\left(1-p_{X\to Y,N}^{*}\right)
\right).
\]
Consequently, an approximate \(100(1-\alpha)\%\) Monte Carlo confidence interval for \(p_{X\to Y,N}^{*}\) is
\[
\widehat p_{X\to Y,S}
\pm
z_{1-\alpha/2}
\sqrt{
\frac{
\widehat p_{X\to Y,S}
\left(1-\widehat p_{X\to Y,S}\right)
}{S}
}.
\]

Now define the population-level outcome probability
\[
p_{X\to Y,N}
=
\mathbb{P}\!\left(
\text{Reject }H_{Y\to X}^{0}
\text{ and fail to reject }H_{X\to Y}^{0}
\right),
\]
where the probability is taken over an i.i.d.\ dataset \(D_N\) of size \(N\). Suppose the bootstrap is consistent for this outcome in the sense that
\[
p_{X\to Y,N}^{*}-p_{X\to Y,N}
\xrightarrow{p}0
\qquad\text{as }N\to\infty.
\]
If \(S=S_N\to\infty\), then
\[
\widehat p_{X\to Y,S_N}-p_{X\to Y,N}
\xrightarrow{p}0.
\]
Furthermore, if \(X\to Y\) is the true causal direction and the test-based procedure is consistent, so that
\[
p_{X\to Y,N}\to 1,
\]
then
\[
\widehat p_{X\to Y,S_N}
\xrightarrow{p}1.
\]
The same results hold separately for the bootstrap rate of each fixed causal discovery outcome.
\end{theorem}

\begin{proof}
Conditional on \(D_N\), the indicators in Equation~\ref{eq:outcome_rate} are i.i.d.\ Bernoulli random variables with success probability \(p_{X\to Y,N}^{*}\). The conditional law of large numbers and central limit theorem therefore yield the stated almost-sure consistency and asymptotic normality as \(S\to\infty\), respectively~\citep{casella2002statistical}.

For the population-level result, write
\[
\widehat p_{X\to Y,S_N}-p_{X\to Y,N}
=
\left(
\widehat p_{X\to Y,S_N}-p_{X\to Y,N}^{*}
\right)
+
\left(
p_{X\to Y,N}^{*}-p_{X\to Y,N}
\right).
\]
Conditional on \(D_N\),
\[
\operatorname{Var}^{*}\!\left(
\widehat p_{X\to Y,S_N}
\mid D_N
\right)
=
\frac{
p_{X\to Y,N}^{*}
\left(1-p_{X\to Y,N}^{*}\right)
}{S_N}
\leq
\frac{1}{4S_N}.
\]
Thus, by Chebyshev's inequality,
\[
\widehat p_{X\to Y,S_N}-p_{X\to Y,N}^{*}
\xrightarrow{p}0
\]
as \(S_N\to\infty\). The second term converges to zero in probability by the assumed bootstrap consistency. Hence,
\[
\widehat p_{X\to Y,S_N}-p_{X\to Y,N}
\xrightarrow{p}0.
\]
Finally, if \(p_{X\to Y,N}\to 1\), then
\[
\widehat p_{X\to Y,S_N}
=
\left(
\widehat p_{X\to Y,S_N}-p_{X\to Y,N}
\right)
+
p_{X\to Y,N}
\xrightarrow{p}1.
\]
\end{proof}

% This frequency can be viewed as an estimate of the probability of each outcome, since it is the sample mean of the corresponding indicator function. Conditional on the observed dataset $D_N$, the bootstrap replicates are i.i.d., so by the conditional law of large numbers and central limit theorem, \eqref{eq:outcome_rate} is consistent and asymptotically normal, allowing for confidence intervals to be constructed using the normal approximation~\citep{casella2002statistical}.

This resampling step quantifies uncertainty in the causal direction by assessing the stability of each outcome across bootstrap replicates. Conceptually, this is a bootstrap procedure, much like those used in the LiNGAM literature \citep[e.g.,][]{komatsu2009, thamvitayakul2012}, in that we repeatedly draw bootstrap datasets and recompute the causal conclusion. In the bivariate setting, both approaches estimate how often each causal outcome would be recovered under repeated sampling. However, the test-based framework generalizes this idea by resampling the data and re-running both hypothesis tests on each bootstrap dataset. Each replicate therefore yields a pair of hypothesis test outcomes (e.g., $\text{Reject $H_{X\to Y}^0$ and fail to reject $H_{Y\to X}^0$ for replication $s$}$) and the resulting bootstrap distribution estimates the joint distribution of these two hypothesis-test outcomes. In contrast, LiNGAM bootstrap procedures resample the data, refit the LiNGAM algorithm on each bootstrap sample, and record only the resulting causal ordering; the bootstrap distribution therefore reflects how often the algorithm selects $X \to Y$ versus $Y \to X$ across replicates. Moreover, LiNGAM bootstrap methods assess the stability of estimated causal coefficients or graph structures within a specific model class, whereas our approach applies the same resampling principle directly to the joint hypothesis-test outcomes. 
% The test-based approach---paired with a functional causal discovery model class (e.g., LiNGAM, ANM, or PNL)---yields uncertainty quantification within a formal inferential framework.

As summarized in Table \ref{tab:fcd-inference-summary}, the test-based approach occupies a favorable position among existing functional causal discovery methods by providing a unified inferential framework that integrates two complementary forms of inference within a hypothesis-testing paradigm. The discussion below elaborates on how these advantages arise.

The test-based approach extends and improves upon using LiNGAM, ANM, and PNL causal discovery algorithms in several important ways. Rather than relying solely on the comparison of dependence measures or test statistics, the test-based approach obtains two $p$-values corresponding to hypothesis tests in each direction. The results of these hypothesis tests convey uncertainty in the direction estimate and provide insights into assumption violations. Although two tests are performed, they correspond to different modeling assumptions---one under $X \to Y$ and another under $Y \to X$---and their outcomes are interpreted jointly. Because these hypotheses are logically coupled, standard multiple-comparison adjustments are neither required nor appropriate; the inferential target is the pattern of rejections rather than the marginal significance of each test. 

When Assumptions~\ref{assump:model}  hold, the two $p$-values are perfectly correlated, meaning that, with enough samples, rejecting in one direction would imply a failure to reject in the other. Therefore, given no assumption violations and a sufficiently large sample size, the test-based approach would provide the same information as a traditional functional causal discovery method. 

However, one cannot be sure that strict assumptions always hold in practice. Under assumption violations, the two $p$-values would not be perfectly correlated. Then, the test-based approach would convey this information by either rejecting or failing to reject in both directions, resulting in an inconclusive directionality (see Table \ref{Tab:sen_assume}). In contrast, many implementations of functional causal discovery methods output a single causal direction (e.g., $X \to Y$) without providing the option for an inconclusive outcome (e.g., \cite{Kalisch2012-pcalg,Ikeuchi2023-lingam,Kalainathan2019-cdt}).

As also reflected in Table \ref{tab:fcd-inference-summary}, previous work (e.g.,~\cite{shimizu2011directlingam, hoyer2009nonlinear, peters2014causal, zhang2012identifiability, mooij_distinguishing_2016}) has not formalized an inferential framework that integrates hypothesis testing with additional inferential quantities to provide statistical guarantees in functional causal discovery. Unlike the bootstrap-based approaches of \citet{komatsu2009} and \citet{thamvitayakul2012} that assess stability under the assumption that the LiNGAM causal model is correctly specified, our framework explicitly incorporates information about model assumptions through formal hypothesis tests. Moreover, unlike the confidence-set approach proposed by \citet{wang2025confidence}, our framework can be paired with the broader Post-Nonlinear (PNL) model class and offers multiple forms of inference including uncertainty quantification---both through hypothesis testing itself and through the estimated rates of causal discovery outcomes. We now turn to numerical results to demonstrate the test-based approach in practice.

\addtolength{\textheight}{-.2in}%

\section{Numerical Results}\label{sec-sims}

In this section, we illustrate the use and behavior of the test-based approach through both simulated and real data. In the simulations studies in Section~\ref{sec-simulations}, we pair the test-based framework with the LiNGAM and PNL model classes and assess its performance in controlled settings with varying degrees of assumption violations. 

For real data analyses in Section~\ref{sec-realdata}, we pair the approach with LiNGAM to demonstrate its application in real settings where one might ordinarily apply LiNGAM-based causal discovery and where the extent of assumption violations is unknown. In this section, we limit our illustrations to the LiNGAM model class. This is because allowing for more extensive PNL or ANM model classes would require additional machine-learning steps to estimate the functions $f_1$ (and $f_2$ in the PNL case). Such model-selection decisions---choosing architectures, tuning hyperparameters, and optimizing nonlinear predictors---are outside the scope of this paper and would distract from our main emphasis on illustrating inferential guarantees provided by the test-based approach.

\subsection{Simulations}\label{sec-simulations}

In our simulations, we pair the test-based approach to two model classes, LiNGAM and PNL. For both classes, the simulations illustrate how the hypothesis-testing formulation in~\eqref{eq:h_indep} can be implemented in practice within a given model class and how the resampling step yields empirical causal outcome rates that quantify uncertainty. However, the scope and emphasis of the simulation studies differ across the two model classes.

Specifically, for the LiNGAM model class, we use simulations to demonstrate how the test-based approach paired with LiNGAM behaves under both model misspecification and unidentified settings. In this case, we further compare the resulting decisions and resampling-based outcome rates to those obtained from DirectLiNGAM with bootstrap rates, as is commonly used in applications (see Table~\ref{tab:fcd-inference-summary}). This comparison highlights differences in the inferential information provided by the two approaches, particularly in settings where model assumptions fail.

For the PNL model class, we focus exclusively on the behavior of the test-based approach under varying degrees of model misspecification. We do not consider unidentified PNL settings, as the linear–Gaussian case already provides the most intuitive and canonical example of non-identifiability in functional causal models, which is examined in the LiNGAM simulations. As a result, the PNL simulations are used to assess how the test-based approach responds to misspecification within a flexible nonlinear model class. In addition, because existing PNL-based methods do not provide a comparable bootstrap-based inferential procedure, we do not include a direct comparison to alternative PNL methods with inference in this setting.

\subsubsection{LiNGAM Model Class Simulations}

When paired with LiNGAM, the hypothesis tests in the test-based approach can be formalized as follows:
\begin{equation}
\label{eq:h_indep_linear}
\begin{aligned}
        H_{X\to Y}^0: X \perp \eta,\; \text{relationship between } X \text{ and } Y \text{ is linear},
        \quad H_{X\to Y}^1: \text{otherwise}, \\
        H_{Y\to X}^0: Y \perp \xi,\; \text{relationship between } Y \text{ and } X \text{ is linear},
        \quad H_{Y\to X}^1: \text{otherwise}.
\end{aligned}
\end{equation}

We view this as a special case of~\eqref{eq:h_indep} in which $f_1$ and $g_1$ are restricted to be linear and $f_2$ and $g_2$ are identity maps, yielding linear structural equations with additive noise.

To study the behavior of the test-based approach paired with LiNGAM and DirectLiNGAM, we consider two forms of assumption violations: (i) varying degrees of linear model misspecification (violation of Assumption A1) and (ii) data-generating processes that span a range of non-Gaussianity settings, including the linear–Gaussian setting, which is unidentified under Assumption A5.

For each simulation setting, we generate \(M = 100\) independent datasets, each of size \(N = 3000\), from the same data-generating process under the true causal direction \(X \rightarrow Y\). For each of the \(M\) simulated datasets, we run the test-based approach to jointly test for independence and goodness-of-fit of the linear model, as formalized in~\eqref{eq:h_indep_linear}, using the combined independence and goodness-of-fit test based on the Hilbert--Schmidt Independence Criterion (HSIC) proposed by \citet{sen_testing_2014}. We compare the results with those from DirectLiNGAM,\footnote{We use DirectLiNGAM rather than ICA-LiNGAM because DirectLiNGAM, under its assumptions, is guaranteed to converge to the correct causal ordering in a finite number of iterations as the sample size tends to infinity \citep{shimizu2011directlingam}. This guarantee does not hold for ICA-LiNGAM. In the bivariate case, DirectLiNGAM fits linear models in both directions and determines the causal direction by assessing the independence between the predictor and its residuals. For each direction, the algorithm computes a user-specified measure of independence and selects the direction that minimizes this measure.} ensuring comparability by also using HSIC as the independence measure within DirectLiNGAM. For each of the \(M\) simulated datasets, we additionally compute the causal outcome rates~\eqref{eq:outcome_rate} for the test-based approach and the bootstrap rates for DirectLiNGAM following the procedure of \citet{thamvitayakul2012}.

\begin{table}[htb]
    \footnotesize
    \centering
    \caption{Simulation settings for varying linearity.}
    \begin{tabular}{p{3.25cm}|p{3cm}|p{6cm}}
    \hline
        Linear & Polynomial = 1 & $Y = \operatorname{sign}(X-a)\lvert X-a\rvert \beta + \eta$\\ \hline
        Moderately nonlinear & Polynomial = 1.5 & $Y = \operatorname{sign}(X-a)\lvert X-a\rvert^{1.5}\beta + \eta$\\ \hline
        Nonlinear & Polynomial = 3  & $Y = \operatorname{sign}(X-a)\lvert X-a\rvert^3\beta + \eta$\\
        \hline
    \end{tabular}
    \label{tab:lin_set_table}
\end{table}

\begin{table}[htb]
    \centering
    \caption{Simulation settings for varying levels of Gaussianity. Here $k$ is the number of mixture components in the Gaussian mixture model (GMM).}
    \begin{tabular}{p{4cm}|p{6cm}}
    \hline
        Gaussian & $X \sim N(0,1), \eta \sim N(\mu_1, \sigma_1)$\\ \hline
        Slightly non-Gaussian & $X \sim N(0,1), \eta \sim \text{GMM}(k=2)$\\ \hline
        Non-Gaussian & $X \sim N(0,1), \eta \sim \text{GMM}(k=3)$ \\
        \hline
    \end{tabular}
    \label{tab:gauss_set_table}
\end{table}

Table~\ref{tab:lin_set_table} summarizes the functional form used to generate \(Y\) from \(X\) under the three levels of linearity assumption violations considered: no violations (where $Y$ is a linear function of $X$), moderate violations (where $Y$ is a polynomial function of $X$ with a polynomial power of 1.5), and strong violations (where $Y$ is a polynomial function of $X$ with a polynomial power of 3). We ensure all other assumptions, besides linearity, are met. We sample $X$ from a truncated exponential distribution (Figure \ref{fig:gd}), which ensures that $Y$ is not normally distributed. Figure \ref{fig:lin} illustrates the simulated data under each linearity setting.

Table~\ref{tab:gauss_set_table} summarizes the noise distributions used to generate \(\eta\), corresponding to three levels of increasing Gaussianity, or equivalently, decreasing degrees of violation of the non-Gaussianity assumption, considered in our simulations: (1) non-Gaussian errors generated from a three-component Gaussian mixture model (GMM), (2) slight violations from a two-component GMM, and (3) Gaussian errors. To ensure a controlled progression toward Gaussianity, mixture parameters were chosen so that the three-component mixture produces a visibly multimodal, asymmetric distribution, while the two-component mixture produces a unimodal but still heavy-tailed and slightly skewed distribution. Specifically, for the three-component case, mixture means $(-2, 0, 2)$ and unequal variances $(0.5, 1, 3)$ with mixture weights $(0.4, 0.2, 0.4)$ yield a clearly non-Gaussian, multimodal shape. For the two-component case, mixture means $(-1, 0.25)$ and variances $(0.5, 0.6)$ with weights $(0.6, 0.4)$ generate a near-unimodal density with heavier tails and mild skewness relative to a standard Gaussian. Finally, under the Gaussian setting, the noise, $\eta$, is drawn from a single standard normal distribution. We ensure that all other assumptions, aside from non-Gaussianity, are satisfied by sampling $X$ from a Gaussian distribution (Figure~\ref{fig:gd}). When both $X$ and the noise $\eta$ are Gaussian, the resulting $Y$ is also Gaussian, placing the model in the linear–Gaussian setting where causal direction is not identified. Figure~\ref{fig:gerr} illustrates the corresponding error distributions under each non-Gaussianity setting.

\begin{figure}[htb]
     \centering
     \begin{subfigure}[b]{0.3\textwidth}
         \centering
         \includegraphics[width=\textwidth]{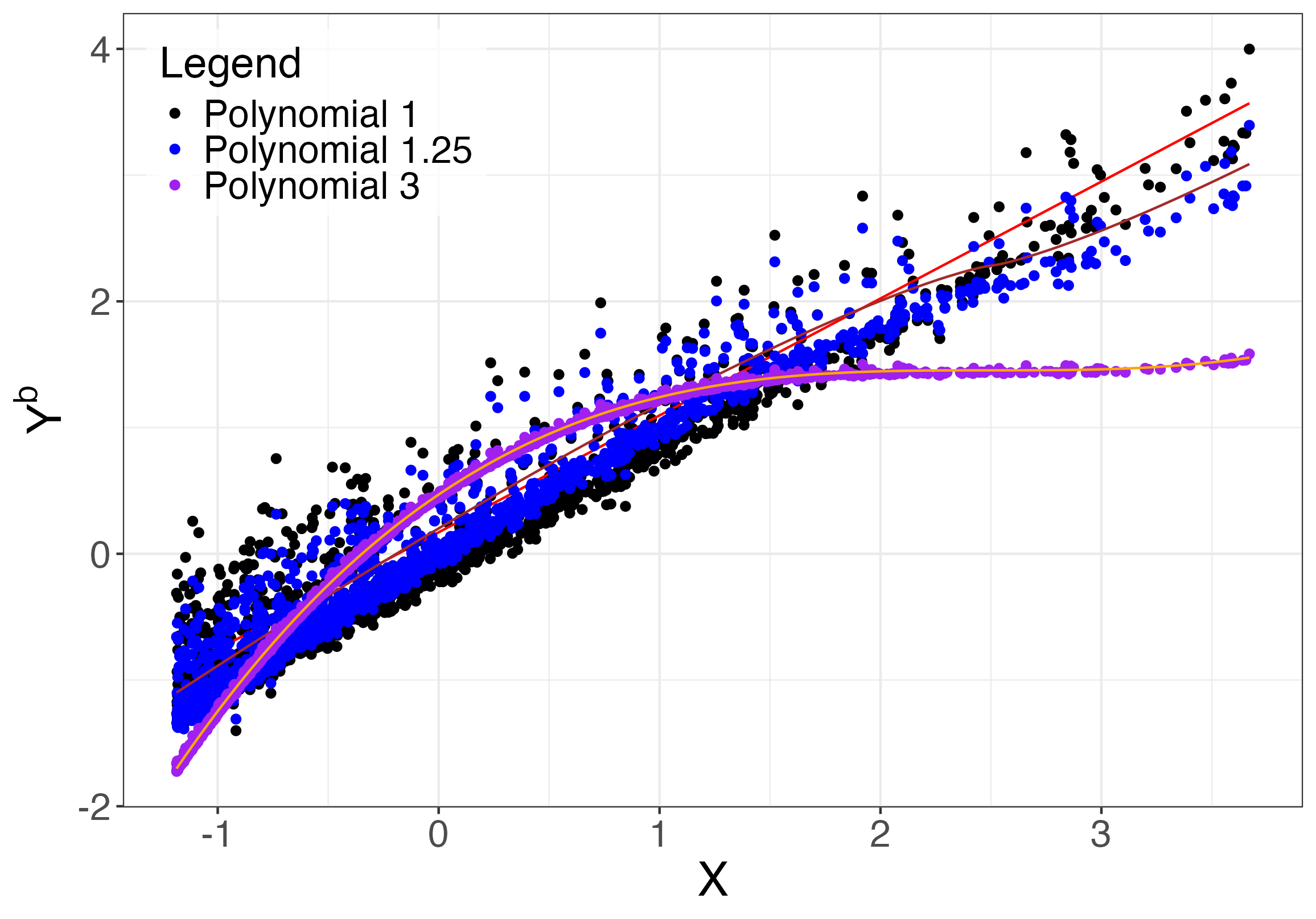}
        \caption{Settings of linearity}
         \label{fig:lin}
     \end{subfigure}
     \hfill
     \begin{subfigure}[b]{0.3\textwidth}
         \centering
         \includegraphics[width=\textwidth]{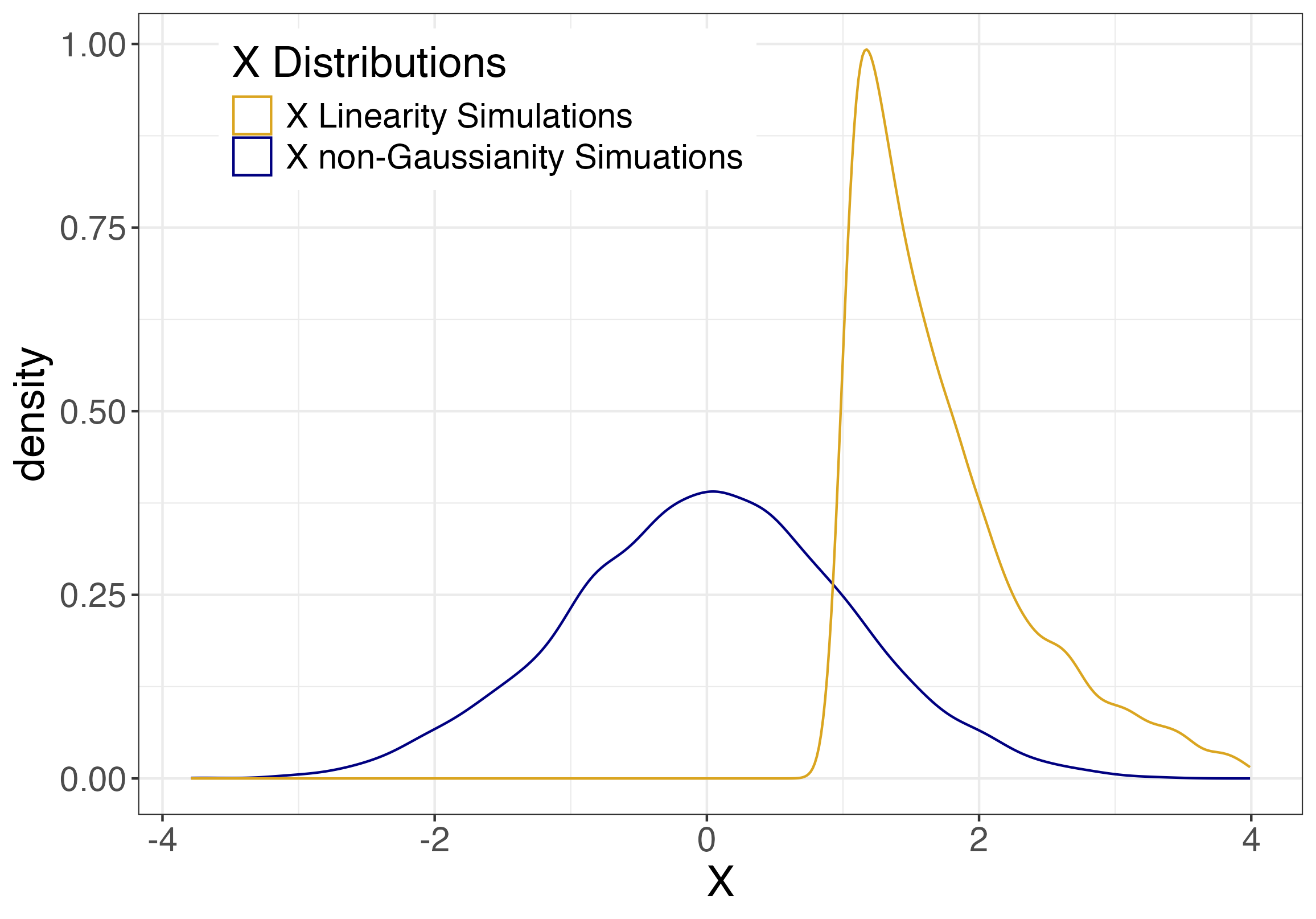}
         \caption{Simulation $X$ distribution}
         \label{fig:gd}
     \end{subfigure}
     \hfill
     \begin{subfigure}[b]{0.3\textwidth}
         \centering
         \includegraphics[width=\textwidth]{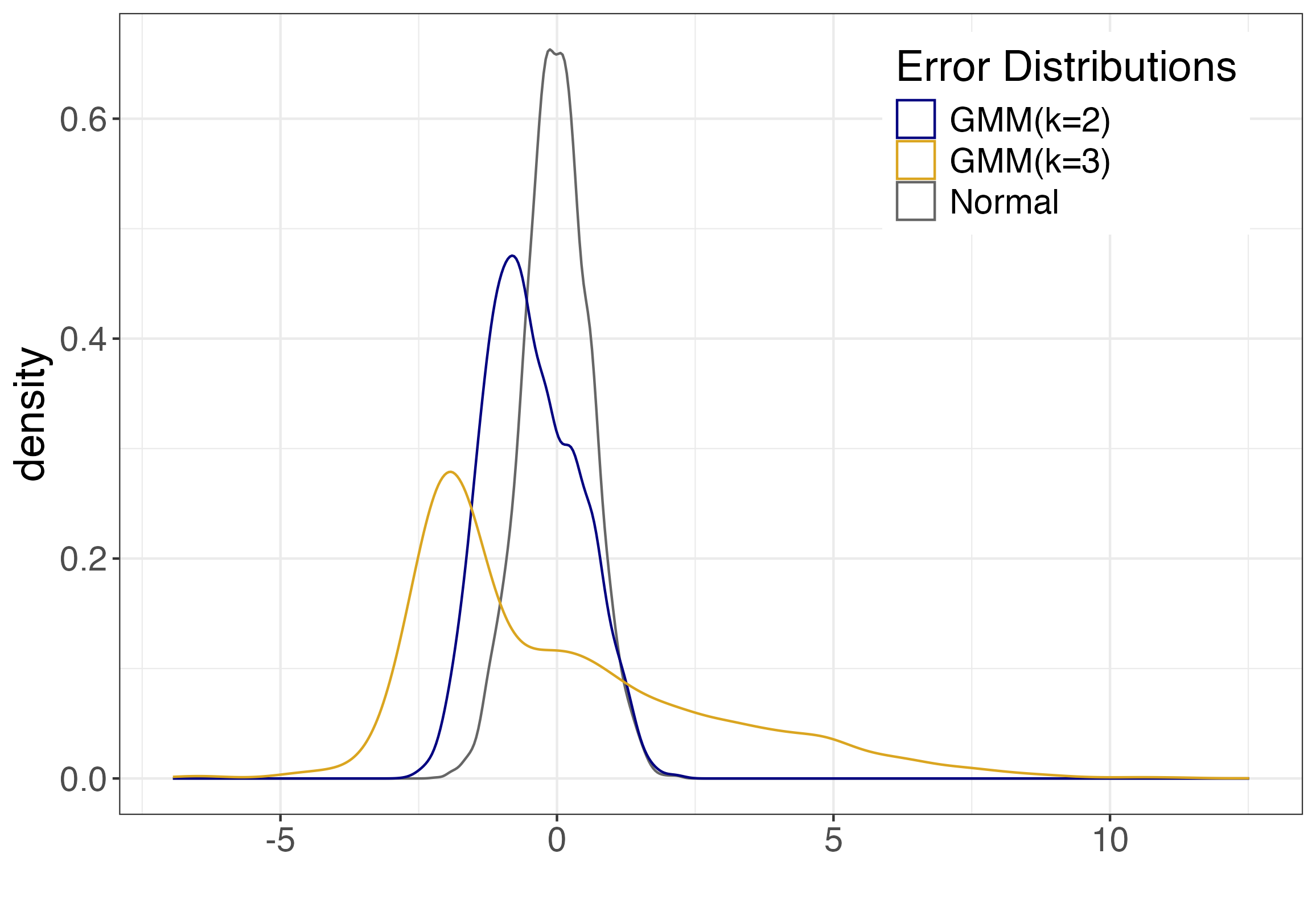}
         \caption{Error $(\eta)$ distributions}
         \label{fig:gerr}
     \end{subfigure}
        \caption{Simulation settings for linearity and Gaussianity.}
        \label{fig:three_graphs}
\end{figure}

\begin{table}[htb]
\scriptsize
\centering
\begin{tabular}{P{2.8cm} P{2.8cm} P{4cm} P{2.4cm} P{2.6cm}}
\toprule
& \multicolumn{2}{c}{Test-based approach} & \multicolumn{2}{c}{DirectLiNGAM} \\
\cmidrule(lr){2-3} \cmidrule(lr){4-5}
Degree of misspecification 
& Decision rate
& Uncertainty quantification 
& Decision rate
& Uncertainty quantification \\
\midrule
Linear (polynomial 1) 
& \makecell[l]{
\textbf{$\mathbf{X \to Y}$: 100\%}\\
$Y \to X$: 0\% \\
Inconclusive: 0\%}
& \makecell[l]{%
Reject both: 3.1\%\\
Fail to reject both: 0\%\\
Reject only $X\!\to\!Y$: 0\%\\
Reject only $Y\!\to\!X$: 96.9\%}
& \makecell[l]{
\textbf{$\mathbf{X \to Y}$: 100\%}\\
$Y \to X$: 0\%} 
& \makecell[l]{%
$X\!\to\!Y$: 99.6\%\\
$Y\!\to\!X$: 0.4\%} \\
\midrule
Moderately nonlinear (polynomial 1.5)
& \makecell[l]{
\textbf{$\mathbf{X \to Y}$: 0\%}\\
$Y \to X$: 0\% \\
Inconclusive: 100\%}
& \makecell[l]{%
Reject both: 100\%\\
Fail to reject both: 0\%\\
Reject only $X\!\to\!Y$: 0\%\\
Reject only $Y\!\to\!X$: 0\%}
& \makecell[l]{
\textbf{$\mathbf{X \to Y}$: 0\%}\\
$Y \to X$: 100\%} 
& \makecell[l]{%
$X\!\to\!Y$: 0\%\\
$Y\!\to\!X$: 100\%} \\
\midrule
Nonlinear (polynomial 3)
& \makecell[l]{
\textbf{$\mathbf{X \to Y}$: 0\%}\\
$Y \to X$: 0\% \\
Inconclusive: 100\%}
& \makecell[l]{%
Reject both: 100\%\\
Fail to reject both: 0\%\\
Reject only $X\!\to\!Y$: 0\%\\
Reject only $Y\!\to\!X$: 0\%}
& \makecell[l]{
\textbf{$\mathbf{X \to Y}$: 0\%}\\
$Y \to X$: 100\%} 
& \makecell[l]{%
$X\!\to\!Y$: 0\%\\
$Y\!\to\!X$: 100\%} \\
\bottomrule
\end{tabular}
\caption{
Comparison between the test-based approach and DirectLiNGAM under varying degrees of linearity assumption violation. Rows correspond to no violation (linear), moderate violation (polynomial exponent $1.5$), and strong violation (polynomial exponent $3$). For the test-based approach, we report the decision and average outcome rates across $M$ simulated datasets. For DirectLiNGAM, we report the decision and average bootstrap rates.
}
\label{tab:lin_sim_res}
\end{table}

Tables~\ref{tab:lin_sim_res} and~\ref{tab:gauss_sim_res} summarize the results for the test-based approach and DirectLiNGAM under each level of linearity and non-Gaussianity assumption violation, respectively. For the test-based approach, we report decision rates across the $M$ replications, along with the average outcome rates aggregated over the $M$ simulated datasets. For DirectLiNGAM, we likewise report decision rates across the $M$ replications and the average bootstrap resampling rates aggregated over the $M$ simulated datasets.

Simulation results under varying degrees of linearity violation show that the test-based approach appropriately signals assumption violations, frequently yielding inconclusive outcomes, whereas DirectLiNGAM tends to select an incorrect causal direction when linearity assumptions are violated. When there are no assumption violations, both DirectLiNGAM and the test-based approach correctly identify the causal direction \(X \to Y\) across all replications. For DirectLiNGAM, the average bootstrap rate for the correct direction is 99.6\%. For the test-based approach, the average outcome rate corresponding to the correct direction is quite high (96.9\%), with all other average causal outcome rates being quite small. Under moderate and strong model misspecification, DirectLiNGAM incorrectly favors the reverse direction \(Y \to X\), with average bootstrap rates of 100\% in both settings. In contrast, across all replications, the test-based approach yields an inconclusive result, as both hypothesis tests are rejected---indicating assumption violations due to model misspecification. The average causal outcome rates further support this conclusion, with the average rate of rejecting both directions equal to 100\% and all other average rates equal to 0\%. Together, the moderate and strong model misspecification settings illustrate how the inferential framework provided by the test-based approach enables a more informative interpretation of uncertainty. Specifically, the test-based approach correctly signals that no causal direction can be reliably inferred under assumption violations and therefore yields an inconclusive result. By contrast, while DirectLiNGAM incorporates a bootstrap-based inferential procedure, it fails to indicate the presence of assumption violations and instead favors an incorrect causal direction due to its enforced deterministic decision rule.

\begin{table}[htb]
\scriptsize
\centering
\begin{tabular}{P{2.8cm} P{2.8cm} P{4cm} P{2.4cm} P{2.6cm}}
\toprule
& \multicolumn{2}{c}{Test-based approach} & \multicolumn{2}{c}{DirectLiNGAM} \\
\cmidrule(lr){2-3} \cmidrule(lr){4-5}
Degree of Gaussianity
& Decision rate
& Uncertainty quantification
& Decision rate
& Uncertainty quantification \\
\midrule
$GMM(k=3)$ 
& \makecell[l]{
\textbf{$\mathbf{X \to Y}$: 100\%}\\
$Y \to X$: 0\% \\
Inconclusive: 0\%}
& \makecell[l]{%
Reject both: 2.7\%\\
Fail to reject both: 0\%\\
Reject only $X\!\to\!Y$: 0\%\\
Reject only $Y\!\to\!X$: 97.3\%}
& \makecell[l]{
\textbf{$\mathbf{X \to Y}$: 100\%}\\
$Y \to X$: 0\%} 
& \makecell[l]{%
$X\!\to\!Y$: 99.1\%\\
$Y\!\to\!X$: 0.9\%} \\
\midrule
$GMM(k=2)$
& \makecell[l]{
\textbf{$\mathbf{X \to Y}$: 70\%}\\
$Y \to X$: 0\% \\
Inconclusive: 30\%}
& \makecell[l]{%
Reject both: 5.2\%\\
Fail to reject both: 24.9\%\\
Reject only $X\!\to\!Y$: 1.1\%\\
Reject only $Y\!\to\!X$: 68.8\%}
& \makecell[l]{
\textbf{$\mathbf{X \to Y}$: 96\%}\\
$Y \to X$: 4\%} 
& \makecell[l]{%
$X\!\to\!Y$: 95.2\%\\
$Y\!\to\!X$: 4.8\%} \\
\midrule
Gaussian
& \makecell[l]{
\textbf{$\mathbf{X \to Y}$: 0\%}\\
$Y \to X$: 0\% \\
Inconclusive: 100\%}
& \makecell[l]{%
Reject both: 1.2\%\\
Fail to reject both: 97.6\%\\
Reject only $X\!\to\!Y$: 0.4\%\\
Reject only $Y\!\to\!X$: 0.8\%}
& \makecell[l]{
\textbf{$\mathbf{X \to Y}$: 45\%}\\
$Y \to X$: 55\%} 
& \makecell[l]{%
$X\!\to\!Y$: 42.9\%\\
$Y\!\to\!X$: 57.1\%} \\
\bottomrule
\end{tabular}
\caption{Comparison between the test-based approach and DirectLiNGAM under varying degrees of non-Gaussianity assumption violation. Rows correspond to no violation (errors from a $GMM(k=3)$), slight violation (errors from a $GMM(k=2)$), and strong violation (Gaussian data and errors i.e. direction is unidentified). For the test-based approach, we report the decision and average outcome rates across $M$ simulated datasets. For DirectLiNGAM, we likewise report the decision and average bootstrap-based stability rates across $M$ simulated datasets.}
\label{tab:gauss_sim_res}
\end{table}

Similarly to the (non-)linearity settings, simulation results under varying degrees of non-Gaussianity assumption violation also show that the test-based approach appropriately signals when the causal direction is unidentified, yielding inconclusive outcomes, whereas DirectLiNGAM tends to oscillate between the two directions when the non-Gaussianity assumption is violated. When there are no or only slight violations of the non-Gaussianity assumption and all other assumptions are satisfied, both DirectLiNGAM and the test-based approach correctly identify the causal direction \(X \to Y\) for the majority of replications (100\% for both under no violations, and 96\% for DirectLiNGAM and 70\% for the test-based approach under slight violations). For DirectLiNGAM, the average bootstrap rates for the correct direction are 99.1\% and 95.2\% under no and slight violations, respectively. Under no assumption violations, the test-based approach yields an average outcome rate of 97.3\% for the correct direction, with all other average causal outcome rates being quite small. Under slight violations, the average rate of correctly identifying the causal direction decreases to 70\%, while the average rate of failing to reject both hypotheses increases to 24.9\%, indicating mild departures from the non-Gaussianity assumption. Under severe assumption violations, when the data are Gaussian, DirectLiNGAM favors the incorrect direction \(Y \to X\) in 55\% of replications, with an associated average bootstrap rate of 57.1\%. Thus, a single deterministic output from DirectLiNGAM is close to a coin flip and may favor the incorrect direction. However, the bootstrap results suggest a more nuanced interpretation: DirectLiNGAM does not strongly support the incorrect direction, but rather yields relatively high bootstrap probabilities for both directions, indicating substantial uncertainty. In this setting, where identifiability breaks down due to Gaussianity, the bootstrap rates can therefore be interpreted as providing a limited form of uncertainty quantification, suggesting an effectively inconclusive result. In contrast, across all replications, the test-based approach yields an inconclusive result, as both hypothesis tests fail to be rejected---indicating assumption violations due to Gaussianity. The average causal outcome rates further support this interpretation, with the rate of failing to reject both hypotheses equal to 97.6\% and all other average rates comparatively small. Together, these settings illustrate how the inferential framework provided by the test-based approach enables a more direct interpretation of uncertainty when identifiability breaks down. Specifically, the test-based approach correctly signals that no causal direction can be reliably inferred under violations of the non-Gaussianity assumption and therefore yields an inconclusive result. By contrast, while the bootstrap in DirectLiNGAM provides a limited form of uncertainty quantification, DirectLiNGAM still enforces a deterministic decision and does not allow for inconclusive outcomes; this behavior may still lead a practitioner to favor the incorrect direction. 

\subsubsection{PNL Model Class Simulations}

When paired with PNLs, we use the hypothesis tests in~\eqref{eq:h_indep}, with $f_1$ and $f_2$ restricted to a specified nonlinear family. 

We use our simulations to study the behavior of the test-based approach paired with PNLs under varying degrees of model misspecification (violation of Assumption A1). We refrain from considering unidentifiable regimes (Assumption A5), since the linear–Gaussian case—the simplest unidentifiable regime—has already been studied for the LiNGAM model class, where we show that the test-based approach correctly signals non-identifiability. Extending these experiments to PNL models would not provide additional insight.

For each model misspecification setting, we generate \(M = 100\) independent datasets, each of size \(N = 3000\), from the same data-generating process under the true causal direction \(X \rightarrow Y\), given by
\[
    Y = f_2\big(f_1(X) + \eta\big),
\]
which corresponds to a post-nonlinear (PNL) causal model. In our simulations, we specify logarithmic \(f_1(x) = \log(x)\) and hyperbolic tangent \(f_2(u) = \tanh(u)\) transformations. 
%\(f_1(x) = \log(x)\) as a logarithmic transformation and \(f_2(u) = \tanh(u)\) as a hyperbolic tangent transformation. 
We use the same marginal distribution for \(X\) and the same error distributions \(\eta\) as in the linearity simulations for LiNGAM model class (see Figures~\ref{fig:gd} and~\ref{fig:gerr}).

To assess the effect of model misspecification, we consider four pairs of fitted models for \(f_1\) and \(f_2\): (1) both functions are correctly specified; (2) \(f_1\) is misspecified as linear while \(f_2\) is correctly specified; (3) \(f_2\) is misspecified as linear while \(f_1\) is correctly specified; and (4) both functions are misspecified as linear. These scenarios allow us to isolate the effects of misspecifying either component of the PNL model and of simultaneously misspecifying both components.

% \(f_1\) and \(f_2\), corresponding to correct specification, partial misspecification of \(f_1\) or \(f_2\), and full misspecification. First, in the correctly specified PNL scenario, the fitted models match the true forms of \(f_1\) and \(f_2\). Second, in the two partially misspecified scenario, we first fit a linear model for \(f_1\) while correctly specifying \(f_2\) as well as a linear model for \(f_2\) while correctly specifying \(f_1\). Third, in a fully misspecified scenario, we fit linear models for both \(f_1\) and \(f_2\). These scenarios allow us to isolate the impact of progressively stronger deviations from the true PNL model on the behavior of the test-based approach.

We find that, across these misspecification settings, the test-based approach behaves analogously to the LiNGAM linear misspecification experiments. When the PNL model is correctly specified, the test-based approach recovers the true causal direction \(X \to Y\) in 100\% of the $M$ replications. Under partial and full misspecification, it yields inconclusive outcomes in 100\% of the $M$ replications, indicating that no reliable causal conclusion can be drawn under these misspecified models. Thus, the PNL simulations corroborate the conclusions from the LiNGAM misspecification experiments: when the model is correctly specified, the test-based approach recovers the correct direction, whereas under misspecification it appropriately signals a lack of reliable evidence through inconclusive outcomes. A table summarizing these results as well as some further discussion of misspecification of $f_2$ is provided in the Appendix.

In summary, our simulation studies confirm that the test-based approach provides a principled and informative inferential framework for functional causal discovery. Across model classes, the simulations show that the test-based approach correctly identifies the causal direction when model assumptions are satisfied. Through hypothesis testing and causal outcome rates, the framework quantifies uncertainty in the inferred direction and yields inconclusive outcomes when the causal direction is unidentified---corresponding to failure to reject both directions---or when model misspecification is present, as indicated by rejection of both directions.

By contrast, the commonly used DirectLiNGAM with bootstrap provides limited inferential information. While DirectLiNGAM with bootstrap summarizes the stability of a single directional decision across resamples, it does not permit inconclusive outcomes or provide insight about potential assumption violations. Under moderate and severe linear model misspecification, DirectLiNGAM favors the incorrect direction, and in the linear–Gaussian setting the causal direction is unidentified, leading DirectLiNGAM to oscillate between the two directions at near 50\% rates. The test-based approach, on the other hand, permits inconclusive outcomes and provides insight into potential assumption violations. In the LiNGAM model class simulations, it is able to signal potential violations of two key assumptions---linearity (Assumption A1) and non-Gaussianity (Assumption A5). The PNL model class simulations further demonstrate that this behavior is not specific to LiNGAM: when PNL assumptions hold, the test-based approach favors the correct causal direction, and when model misspecification is present, it shifts toward inconclusive outcomes rather than enforcing a potentially misleading decision. 

Our simulation results motivate the analysis of real-world datasets, where the data-generating mechanism and the validity of modeling assumptions are typically unknown.

\subsection{Real Data}\label{sec-realdata}

% use pair 0033, near linear and literature on this (will do 0098 instead and not standardize the data)

% use pair 0050 ozone and temperature
We apply the test-based approach to three real-world datasets from the Cause–Effect Pairs Benchmark (Tübingen pairs)~\citep{mooij_distinguishing_2016}: two Balltrack datasets and the Ozone and Temperature dataset. The Tübingen pairs, where the direction of causality is presumed known, are widely used to study causal discovery methods. 
% and have also been used to assess the goodness of fit of linear non-Gaussian structural equation models~\citep{schkoda2025goodness}. 
We selected these three examples because their scatterplots suggested approximately linear relationships. Therefore,
%a regime in which linear causal discovery methods are often presumed appropriate in practice. As a result, 
one might reasonably expect approaches based on the LiNGAM model class to be well suited for these data, making them informative test cases for illustration. 

Through this analysis, we demonstrate how the test-based approach paired with LiNGAM provides a richer inferential framework than DirectLiNGAM. Specifically, separate hypothesis tests for each causal direction allow
%the test-based approach yields separate hypothesis test outcomes for each causal direction, allowing 
practitioners to distinguish between confident directional conclusions and inconclusive cases due to possible assumption violations---insights that are not available from DirectLiNGAM, even when supplemented with bootstrap rates~\citep{thamvitayakul2012}. 
%as is commonly done in applications (see Table~\ref{tab:fcd-inference-summary}).

\begin{figure}[!ht]
    \centering
    \includegraphics[width=160mm]{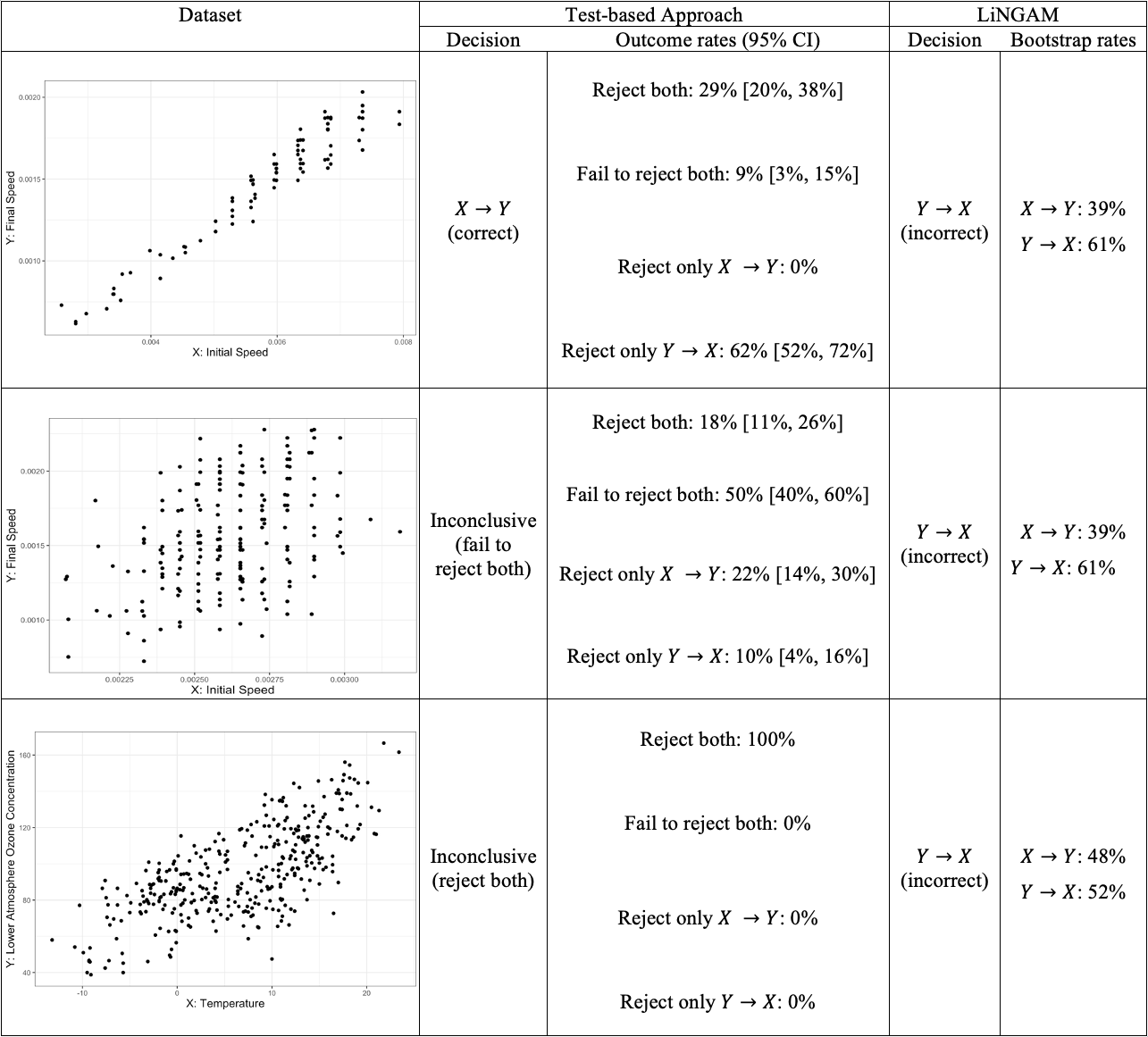}
    \caption{The first two rows presents results for the Balltrack datasets, and the third row presents results for the Temperature and Ozone dataset. From left to right, the columns display the scatterplot of the data, the test-based decision and outcome rates (resampling proportions with 95\% confidence intervals; CIs are omitted when the estimated variance is zero), and the corresponding LiNGAM decision with associated bootstrap rates.}
    \label{fig:real_res}
\end{figure} 

The first Balltrack dataset contains 94 measurements collected using a physical ball track equipped with two pairs of light barrier sensors: the first sensor ($X$) records the ball’s initial speed, and the second ($Y$) records its speed at a later position along the track. The experiment was designed and conducted for the Tübingen Cause–Effect Pairs benchmark suite~\citep{mooij_distinguishing_2016}. The initial portion of the track has a steep slope, making the initial speed highly sensitive to the exact release position. To introduce natural variability and avoid systematic positioning, some runs used positions chosen by an adult, while the remainder used positions chosen by a four-year-old child. Following \citet{mooij_distinguishing_2016}, we treat the causal direction as known, with the initial speed causing the final speed. A scatterplot of initial versus final speed (Figure~\ref{fig:real_res}) suggests an approximately linear relationship, consistent with the diagnostic reported by \citet{schkoda2025goodness}. Nevertheless, the DirectLiNGAM point estimate infers the incorrect causal direction when applied to this dataset, and the associated bootstrap procedure supports this incorrect direction in 61\% of resamples. In contrast, the test-based approach supports the correct causal direction. Specifically, the causal outcome rate for the correct direction is 62\%, while the remaining outcome rates indicate moderate assumption violations, potentially due to model misspecification, as evidenced by a 29\% rate of rejecting both directions. This example illustrates that relying on the DirectLiNGAM point estimate alone is not advisable and that, even when supplemented with bootstrap rates, it can still be misleading. By contrast, the test-based approach not only recovers the correct causal direction but also provides diagnostic information about potential assumption violations.
%, offering a more informative inferential summary. 

The second Balltrack dataset contains 202 measurements collected under a setup similar to the first Balltrack dataset with two sensors $X$ and $Y$ recording the initial and a later speed along the track. 
%collected using a physical ball track equipped with two pairs of light barrier sensors: the first ($X$) records the ball’s initial speed, and the second ($Y$) records its speed at a later position along the track. 
However, in this setup, the track had a shorter acceleration zone, which limited the extent to which the ball can increase its speed. As with the first Balltrack dataset, the experiment was designed and conducted for the Tübingen Cause–Effect Pairs benchmark suite \citep{mooij_distinguishing_2016}, and we
%Following \citet{mooij_distinguishing_2016}, we 
treat the causal direction as known, with the initial speed causing the final speed. 
A scatterplot of initial versus final speed (Figure~\ref{fig:real_res}) suggests an approximately linear relationship, consistent with the diagnostic reported by \citet{schkoda2025goodness}. Here, the DirectLiNGAM point estimate infers the incorrect causal direction when applied to this dataset, and the associated bootstrap procedure supports this incorrect direction in 61\% of resamples. By contrast, the test-based approach yields an inconclusive result by failing to reject either causal direction, with a causal outcome rate of 50\%. This indicates potential identifiability or small sample size issues, an informative finding which is preferable to LiNGAM’s enforced deterministic (and incorrect) decision. % Since non-identifiability due to exact Gaussianity can be rare in real data, this outcome more likely reflects limited sample size, suggesting that this dataset may benefit from additional observations. This outcome is preferable to LiNGAM’s enforced deterministic (and incorrect) decision, as it not only avoids a misleading conclusion but also signals the potential need for more data.

The Ozone and Temperature dataset contains 365 daily mean temperatures and lower-atmosphere ozone concentrations collected during 2009 in Chaumont, Switzerland, by the Swiss Federal Office for the Environment~\cite{adminAirData}. Following~\cite{mooij_distinguishing_2016}, we assume changes in temperature cause changes in lower-atmosphere ozone concentrations. The relationship between the two variables appears approximately linear (see the scatterplot in Figure~\ref{fig:real_res}).
%, consistent with the results of \citet{schkoda2025goodness}, who find that a linear model becomes tenable after allowing for one latent confounder, consistent with the results of \citet{schkoda2025goodness}, who find that a linear model becomes tenable after allowing for one latent confounder. 
Applying DirectLiNGAM to this dataset yields the incorrect causal direction, with a bootstrap rate of 52\% (Figure~\ref{fig:real_res}). Because the two directions receive nearly equal bootstrap support, this rate provides a limited form of uncertainty quantification and suggests that the LiNGAM's result is effectively inconclusive. The test-based approach yields an unequivocally inconclusive outcome by rejecting both causal directions, indicating a potential violation of the modeling assumptions. Although our procedure does not identify the violated assumption, diagnostics of \citet{schkoda2025goodness}---who found that a linear model becomes tenable for the Ozone and Temperature dataset after allowing for one latent confounder---suggest that unobserved confounding is one possible explanation. The test-based approach's inconclusive outcome is preferable to the DirectLiNGAM's result because it is not only more explicit but also provides diagnostic information about a possible violation of the modeling assumptions.

Overall, the real-data analyses further illustrate how the test-based approach provides a richer inferential framework than DirectLiNGAM with bootstrap. In particular, it produces directional conclusions when supported by the data and inconclusive outcomes otherwise; in addition, it provides outcome rates and confidence intervals that quantify uncertainty across resamples. This distinction---between enforcing a single deterministic direction and allowing uncertainty-aware, inconclusive outcomes---is particularly important in applied settings, where the data-generating mechanisms and the extent of assumption violations are typically unknown. In such scenarios, inconclusive outcomes are especially informative: rejecting both causal directions can signal potential model misspecification, while failing to reject either direction may indicate lack of identifiability or insufficient sample size under the assumed model class. By permitting and interpreting these distinct inconclusive outcomes, the test-based approach provides diagnostic insight that is not explicitly available from LiNGAM-based causal discovery methods.

% \begin{itemize}
%     \item Fully linear data: show that tests recover the correct direction under model assumptions.
%     \item Slight violations of linearity
%     \item Non-identifiable Gaussian data
%     \item Use small, clear tables/figures to emphasize interpretation over computation.
% \end{itemize}

% \section{Practical Lessons for Statisticians (draft by EOD 9/29 to discuss at 10/6 and 10/20 meetings)}\label{sec-practical-lessons}
% % \begin{itemize}
% %     \item practical lessons for statisticians 
% % \end{itemize}

\section{Discussion}\label{sec-conc}

Causal discovery methods have pioneered data-driven approaches for identifying and understanding causal relationships, a goal that is fundamental across many scientific disciplines. However, applied researchers using causal discovery methods should be mindful of the (lack of) statistical guarantees in existing approaches such as LiNGAM, ANM, or PNL. 

In this paper, we reviewed the state of statistical guarantees provided by existing functional causal discovery methods. As summarized in Table~\ref{tab:fcd-inference-summary}, this review revealed a clear gap: the absence of a formal, general inferential framework that integrates multiple complementary forms of inference while making both assumptions and uncertainty explicit. As a first step toward filling this gap, we developed a test-based approach to bivariate functional causal discovery that combines two complementary forms of inference---hypothesis testing and uncertainty quantification for each competing causal discovery outcome---within a formal framework, with theoretical properties of the resulting outcome rates established in Theorem~\ref{thm:bootstrap_outcome_rate}. Together, these components provide a principled way to quantify uncertainty in the estimated causal direction and provide diagnostic insight for cases when a determinate causal direction is not supported by the data.

Specifically, the inferential framework of the test-based approach enables practitioners to (i) report evidence supporting a causal direction when such conclusion is supported by the data, (ii) quantify uncertainty in that conclusion, and (iii) obtain diagnostic insight into potential assumption violations---such as model misspecification (Assumption~A1) and non-identifiability (Assumption~A5)---particularly when the causal direction is inconclusive. Importantly, these forms of inference are available across major functional causal model classes. As a result, the test-based approach to functional causal discovery addresses a gap in the existing literature, where this combination of qualities---uncertainty quantification, diagnostic insight, and applicability across model classes---is largely unavailable (see Table~\ref{tab:fcd-inference-summary}).

The generality of our framework derives from its ability to be paired with any suitable procedure for jointly assessing goodness-of-fit and independence between the predictor and the error. The practical guarantees of a particular implementation, however, depend on the theoretical properties of the test used. To our knowledge, no goodness-of-fit and independence test has been developed specifically for the PNL model class. We therefore use the test of \citet{sen_testing_2014}, which was originally developed for additive noise models. Our simulations provide empirical evidence that under the PNL misspecification settings considered, the test rejects both causal directions due to model misspecification, including when the inner function, the outer function, or both are misspecified. Nevertheless, these results do not constitute a general theoretical guarantee for PNL models. Developing tests specifically for the PNL setting, together with corresponding theoretical guarantees, is therefore an important direction for future methodological research.

The lesson for applied statisticians is threefold. First, causal discovery should always be undertaken with explicit awareness of its modeling assumptions and of the consequences when those assumptions fail. Second, inconclusive results or signs of assumption violation should not be viewed as failures but as statistically meaningful signals that can guide model refinement, data collection, or reformulation of causal hypotheses. Third, causal discovery should be paired with explicit uncertainty quantification, and users should understand the nature and scope of the inferential information provided by these tools. The test-based approach represents a first step toward incorporating these key principles into a coherent inferential framework.

Within the test-based approach, we primarily discussed the impact of model misspecification and non-identifiability (e.g., the linear–Gaussian setting). An important direction for future work is to extend this framework to characterize how other types of assumption violations—such as unobserved confounding, cyclicity, and non-i.i.d.\ data—affect the behavior of the two hypothesis tests in~\eqref{eq:h_arrow} and the resulting causal discovery outcomes. Understanding how these violations, individually or in combination, influence test outcomes would further enhance the diagnostic scope and practical utility of the test-based approach.

% While we studied the test-based approach in the bivariate setting in this paper, extending the framework to multivariate settings is another important direction for future work. One very simple multivariate setting in which the framework can directly apply is when the causal structure is known for all variables except for two variables of interest, $X$ and $Y$, and the effects of all other variables on $X$ and $Y$ are understood. In this case, the framework can be used to study the directionality between $X$ and $Y$ with only minor modifications. Specifically, one could first remove the effects of the known confounders on $X$ and $Y$, respectively, via, e.g., linear regression. The resulting residuals $r_X$ and $r_Y$ can then be treated as a new pair of bivariate variables to be directly used in the test-based approach. 
While we studied the test-based approach in the bivariate setting in this paper, extending the framework to multivariate settings is another important direction for future work. One relatively simple multivariate setting arises when the causal structure is known except for the direction between two variables of interest, $X$ and $Y$, and the effects of the other relevant variables on $X$ and $Y$ can be modeled. One might first remove these effects from $X$ and $Y$, following the residualization strategy used in DirectLiNGAM~\citep{shimizu2011directlingam}, and then apply the bivariate test-based approach to the resulting residuals $r_X$ and $r_Y$. However, because $r_X$ and $r_Y$ are estimated from the same data, they cannot simply be treated as observed i.i.d. pairs. Although the test of \citet{sen_testing_2014} accounts for the estimation of the residuals from the regression being tested, this preliminary residualization introduces an additional source of estimation uncertainty. A formal extension would therefore need to propagate this uncertainty through the subsequent goodness-of-fit and independence tests and establish the validity of the resulting procedure.

Extending the test-based approach to a general multivariate setting in which the causal structure among variables is unknown requires substantial methodological development. Applying the bivariate procedure separately to each pair need not produce a globally consistent causal graph. For example, with three variables, the pairwise analyses could yield $X_1 \to X_2$, $X_2 \to X_3$, and $X_3 \to X_1$, which is incompatible with a directed acyclic graph. They may also yield a mixture of directional and inconclusive results. A multivariate extension would therefore require a principled method for aggregating the dependent pairwise tests while controlling error and ensuring global graph compatibility. One possible way to ensure global compatibility would be to construct a sequential procedure analogous to DirectLiNGAM that identifies upstream variables and thereby produces a causal ordering. However, adapting this strategy to the test-based framework is not straightforward. At each stage, identifying an upstream variable would require multiple goodness-of-fit and residual-independence tests. These tests could identify multiple candidates, no candidate, or only inconclusive results, making it unclear which variable should be selected next or whether the procedure should continue. Moreover, because the regressions and tests at later stages depend on earlier ordering decisions, errors or inconclusive outcomes could propagate through the remaining stages. The sequence of tests, the rules for resolving conflicting or inconclusive outcomes, and the resulting error guarantees would therefore all need to be developed. These challenges become more pronounced as the number of variables increases. With $p$ variables, each variable has $2^{p-1}$ possible parent sets before structural restrictions are imposed, and the associated tests are dependent because they use the same data and related regression models. Accounting for additional variables as potential confounders further requires estimating and removing multiple nuisance components. Developing a scalable procedure that produces globally coherent causal structures while providing valid uncertainty quantification therefore remains an important direction for future research.

Finally, the inferential foundations established by the test-based functional causal discovery approach can be extended using ideas from the broader statistical inference literature---such as methods for power analysis \citep{lehmann2005testing,casella2002statistical}, resampling \citep{EfroTibs93,good2005permutation}, and model diagnostics \citep{cook1977detection,atkinson1985plots,buja2009statistical}---to
%---may prove valuable when adapted to the causal discovery setting. 
further strengthen causal discovery methodology and promote its reliable use in practice.

\section{Disclosure statement}\label{disclosure-statement}

The authors have no conflicts of interest to declare.

\section{Acknowledgements}\label{acknowledge}
The authors thank participants of the University of Washington's Statistical and Machine Learning Approaches for the Social Sciences working group for their feedback and support. Shreya Prakash received partial support from the University of Washington's Center for Statistics and the Social Sciences (CSSS) and the Department of Statistics. The authors used OpenAI’s ChatGPT (GPT-5.1 and GPT-5.2) for localized, minor language improvements in some paragraphs in the process of drafting the paper. The tool was not used to generate original research content, data analysis, to assist with literature review or scientific conclusions. All intellectual and scholarly contributions in this manuscript were made by the authors.

\section{Data Availability Statement}\label{data-availability-statement}

The data that support the findings of this study are openly available from the Tübingen Cause–Effect Pairs repository hosted by the Max Planck Institute for Intelligent Systems at \url{https://webdav.tuebingen.mpg.de/cause-effect/}. The datasets are described in detail in \cite{mooij_distinguishing_2016} and are cited in this paper in accordance with the repository’s citation guidelines.

\bibliographystyle{plainnat}
\bibliography{bibliography}

\clearpage
\appendix

\section*{Appendix}

\section{Unidentifiable Cases in Post-Nonlinear Models}
\label{sup:unidentifiable_cases}

Following \citet{zhang2012identifiability}, the causal direction in the Post-Nonlinear (PNL) model is unidentifiable in the following five cases:

\begin{enumerate}
  \item Both $f_1$ and $f_2$ are linear and the noise $\eta$ is Gaussian.
  \item The model reduces to an invertible linear transformation with symmetric distributions for $X$ and $\eta$.
  \item The noise $\eta$ and the cause $X$ have certain exponential family relationships that render $f_1$ and $f_2$ indistinguishable.
  \item The functional relationships satisfy specific differential equation constraints that make both directions equally valid.
  \item Degenerate cases where $f_2$ is constant or non-invertible, leading to lack of model identifiability.
\end{enumerate}

\section{Example Behaviors of
  \texorpdfstring{$\xi$}{xi}}
\label{sup:errors}

To understand the behavior of $\xi$ under the ANM model class ($f_2, g_2$ are identity functions), we examine several illustrative cases.

\begin{enumerate}
    \item Suppose $f_1$ is additive, i.e. $f_1(a+b) = f_1(a)+f_1(b)$, and $g_1 = f_1^{-1}$: $$\xi = f_1^{-1}(Y)-f_1^{-1}(\eta)-f_1^{-1}(Y) = f_1^{-1}(\eta),$$ here $\xi$ is only a function of $\eta$.
    
    \item Suppose \(\eta = 0\), and \(g_1 = f_1^{-1}\). Then
    \[
    \xi = f_1^{-1}(Y) - f_1^{-1}(Y) = 0,
    \]
    so the reverse-direction residual vanishes.
    
    \item Suppose \(f_1\) is nonlinear, and \(g_1\) is restricted to be linear, \(g_1(Y) = \gamma Y\) for some \(\gamma \in \mathbb{R}\). Then
    \[
    \xi = f_1^{-1}(Y - \eta) - \gamma Y.
    \]
    In this case, \(\xi\) depends on both \(Y\) and \(\eta\), and hence \(Y \not\perp \xi\).
    % \item When $f_2, g_2$ are identity functions, $\eta = 0$ and $g_1 = f_1^{-1}$: $$\xi = 0$$
    % \item When $f_2, g_2$ are identity functions, nonlinear $f_1$ but linear $g_1(Y) = \gamma Y, \gamma 
    % \in \mathbb{R}$:
    % $$\xi = f_1^{-1}(Y-\eta)-g_1(Y) = f_1^{-1}(Y-\eta)-(\gamma Y),$$ here $\xi$ contains both the residual term of $Y$ in addition to $\eta$. Thus $Y \not\perp \xi$.
    \end{enumerate}

Taken together, these cases illustrate that under the ANM model class, when $f_1$ is nonlinear, there is no choice of $g_1$ within a restricted model class that yields $\xi \perp Y$.. Consequently, except in the standard unidentifiable cases (e.g., the linear–Gaussian setting), we have $Y \not\perp \xi$ when fitting the model in the incorrect causal direction.

\section{PNL Model Misspecification Simulation Results}

We examine the test-based approach under four PNL model specifications: (1) both $f_1$ and $f_2$ are correctly specified; (2) $f_1$ is misspecified as linear while $f_2$ is correctly specified; (3) $f_2$ is misspecified as linear while $f_1$ is correctly specified; and (4) both functions are misspecified as linear.

The test of \citet{sen_testing_2014} is formulated for an additive model. When the outer transformation is correctly specified and treated as fixed, the PNL model
\[
Y=f_2(f_1(X)+\eta),
\qquad X\perp\eta,
\]
can be transformed into
\[
f_2^{-1}(Y)=f_1(X)+\eta.
\]
Provided that $f_1$ belongs to the regression-function class considered by \citet{sen_testing_2014} and their remaining regularity conditions hold, the transformed model satisfies their null hypothesis, and their bootstrap-consistency result provides asymptotic Type I error control.

This argument does not directly apply when $f_2$ is misspecified. Applying an incorrectly specified inverse transformation may result in a misspecified regression function, dependence between the predictor and error, or both. Such departures correspond to alternatives considered by \citet{sen_testing_2014}, for which their test is asymptotically consistent under their regularity conditions. However, their theory does not characterize when misspecification of the PNL outer transformation necessarily induces one of these departures from the additive-model null. A general theoretical analysis of misspecified outer transformations is beyond the scope of this paper. We therefore investigate this setting empirically by fitting $f_2$ as linear while correctly specifying $f_1$.

Table \ref{tab:pnl_sim_res} presents the simulation results for the test-based approach paired with the PNL model class under the four model specifications. 

\begin{table}[ht]
\scriptsize
\centering
\begin{tabular}{P{4cm} P{3.2cm} P{5.2cm}}
\toprule
& \multicolumn{2}{c}{Test-based approach} \\
\cmidrule(lr){2-3}
Degree of misspecification 
& Decision rates
& Average causal outcome rates \\
\midrule
No misspecification 
& \makecell[l]{%
\textbf{$X \to Y$: 100\%}\\
$Y \to X$: 0\% \\
Inconclusive: 0\%}
& \makecell[l]{%
Reject both: 1.9\%\\
Fail to reject both: 0\%\\
Reject only $X\!\to\!Y$: 0\%\\
Reject only $Y\!\to\!X$: 98.1\%}
\\
\midrule
Partial misspecification ($f_1$)
& \makecell[l]{%
$X \to Y$: 0\%\\
$Y \to X$: 0\% \\
\textbf{Inconclusive: 100\%}}
& \makecell[l]{%
Reject both: 100\%\\
Fail to reject both: 0\%\\
Reject only $X\!\to\!Y$: 0\%\\
Reject only $Y\!\to\!X$: 0\%}
\\
\midrule
Partial misspecification ($f_2$)
& \makecell[l]{%
$X \to Y$: 0\%\\
$Y \to X$: 0\% \\
\textbf{Inconclusive: 100\%}}
& \makecell[l]{%
Reject both: 100\%\\
Fail to reject both: 0\%\\
Reject only $X\!\to\!Y$: 0\%\\
Reject only $Y\!\to\!X$: 0\%}
\\
\midrule
Full misspecification
& \makecell[l]{%
$X \to Y$: 0\%\\
$Y \to X$: 0\% \\
\textbf{Inconclusive: 100\%}}
& \makecell[l]{%
Reject both: 100\%\\
Fail to reject both: 0\%\\
Reject only $X\!\to\!Y$: 0\%\\
Reject only $Y\!\to\!X$: 0\%}
\\
\bottomrule
\end{tabular}
\caption{
Test-based approach under varying degrees of model misspecification. Rows correspond to settings with no misspecification (both $f_1$ and $f_2$ match their true forms), partial misspecification of $f_1$ (only $f_1$ is misspecified as linear), partial misspecification of $f_2$ (only $f_2$ is misspecified as linear), and full misspecification (both $f_1$ and $f_2$ are misspecified as linear). We report the resulting decision rates and average causal outcome rates across $M$ simulated datasets.
}
\label{tab:pnl_sim_res}
\end{table}

When both functions are correctly specified, the test-based approach selects the true direction $X\to Y$ for all $M$ simulated datasets. Across bootstrap samples, the correct directional outcome occurs at an average rate of $98.1\%$, while both directional models are rejected at an average rate of $1.9\%$.

When $f_1$, $f_2$, or both functions are misspecified as linear, both directional null hypotheses are rejected for all $M$ simulated datasets, resulting in an inconclusive decision. In particular, misspecifying $f_2$ produces results similar to those obtained when $f_1$ is misspecified. This is the appropriate diagnostic outcome in these settings because neither fitted directional model adequately represents the data-generating mechanism. These results provide empirical evidence that the procedure is sensitive to this particular form of outer-transformation misspecification, but they do not establish that every form of $f_2$ misspecification necessarily induces a departure from the additive-model null.

\section{Simulation Results Across Sample Sizes}

Following the methodology of \citet{prakash2024diagnostic}\footnote{See \url{https://github.com/shreyap18/causalDiagnose} for the code.}, we carry out simulations examining DirectLiNGAM and the test-based approach at sample sizes ranging from 20 to 1500 under the linearity and non-Gaussianity settings considered in Section~4.1.1 of the main text. For DirectLiNGAM, we report the rates of selecting the correct direction, $X\to Y$ (orange), and the incorrect direction, $Y\to X$ (blue). For the test-based approach, we additionally report the rates of rejecting both directional null hypotheses (purple) and failing to reject both directional null hypotheses (green).

 \begin{figure}
    \centering
    \includegraphics[width=\linewidth]{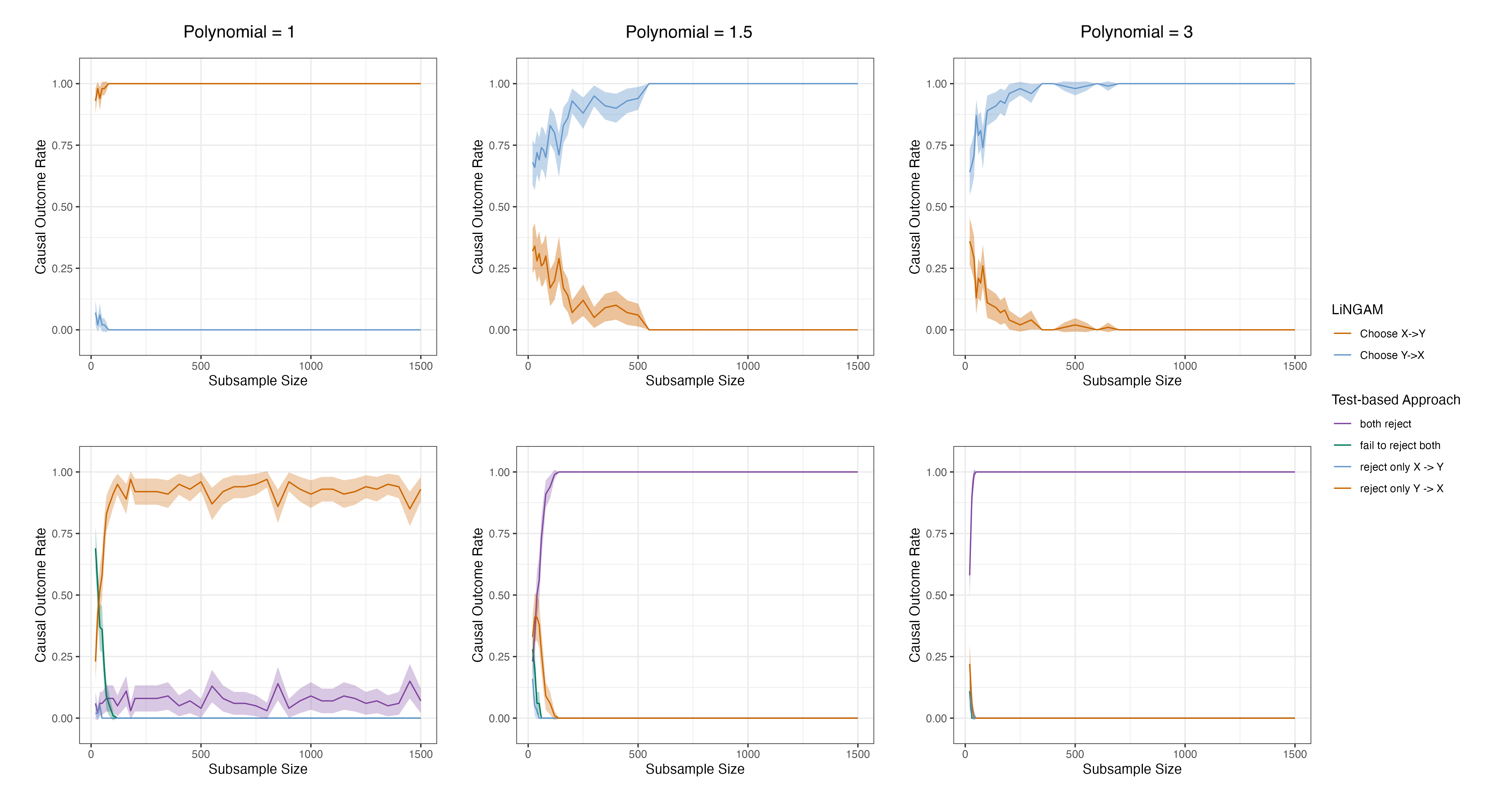}
    \caption{DirectLiNGAM and test-based approach outcome rates across sample sizes under varying degrees of violation of the linearity assumption. In each panel, the \(x\)-axis represents the sample size, up to \(N=1500\), and the \(y\)-axis represents the causal outcome rate. The first row presents the DirectLiNGAM results, and the second row presents the test-based approach results. The columns correspond to no violation (\(Y\) is a linear function of \(X\)), a moderate violation (\(Y\) is a polynomial function of \(X\) with exponent \(1.5\)), and a severe violation (\(Y\) is a polynomial function of \(X\) with exponent \(3\)), respectively. The shaded regions represent pointwise confidence intervals.}
    \label{fig:linearity_sample_size}
\end{figure}

Figure~\ref{fig:linearity_sample_size} presents the results under violations of linearity. When there is no violation, both DirectLiNGAM and the test-based approach support the correct direction with probability approaching one as the sample size increases. Under moderate and severe violations, however, DirectLiNGAM increasingly supports the incorrect direction. In contrast, the test-based approach increasingly rejects both directional null hypotheses, indicating that neither directional model adequately represents the data-generating mechanism.

 \begin{figure}
    \centering
    \includegraphics[width=\linewidth]{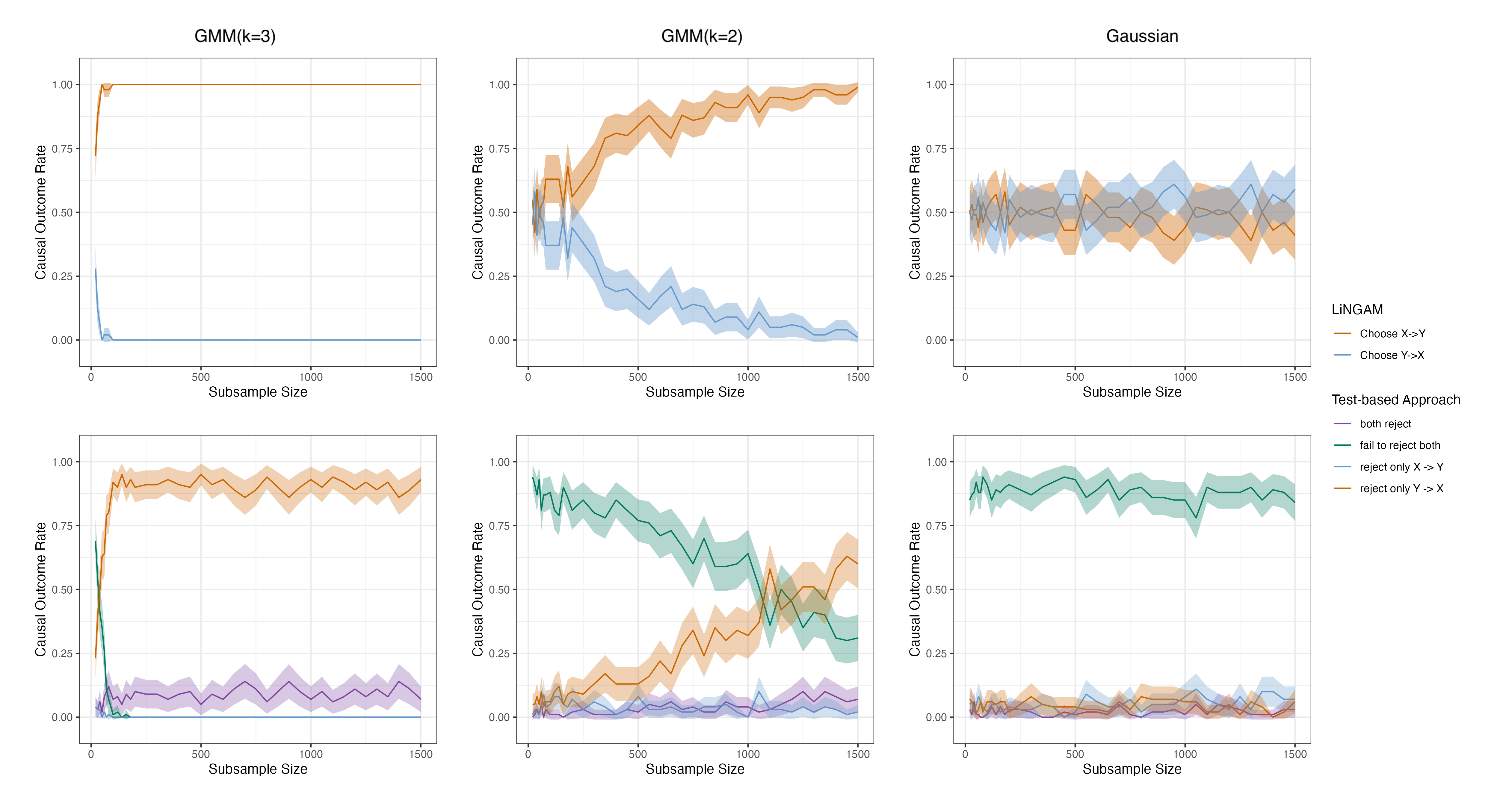}
    
    \caption{DirectLiNGAM and test-based approach outcome rates across sample sizes under varying degrees of violation of the non-Gaussianity assumption. In each panel, the \(x\)-axis represents the sample size, up to \(N=1500\), and the \(y\)-axis represents the causal outcome rate. The first row presents the DirectLiNGAM results, and the second row presents the test-based approach results. The columns correspond, respectively, to no violation (errors from a three-component Gaussian mixture model), a slight violation (errors from a two-component Gaussian mixture model), and a severe violation (a Gaussian predictor and Gaussian errors). The shaded regions represent pointwise confidence intervals.}
    \label{fig:nongaussianity_sample_size}
\end{figure}

Figure~\ref{fig:nongaussianity_sample_size} presents the results under violations of non-Gaussianity. When there is no violation, both methods increasingly support the correct direction as the sample size grows. Under a slight violation, DirectLiNGAM initially supports both directions at approximately equal rates but increasingly favors the correct direction as the sample size increases. For the test-based approach, the rate of failing to reject both null hypotheses decreases with sample size, while the rate of selecting the correct direction exceeds $0.5$ as the sample size increases; the remaining outcome rates stay near zero. Under a severe violation, DirectLiNGAM supports the two directions at approximately equal rates across sample sizes. In contrast, the test-based approach increasingly fails to reject both directional null hypotheses, consistent with non-identifiability under Gaussianity.

\end{document}